\documentclass[authoryear,preprint,3p]{elsarticle}
\usepackage{amsmath,amssymb,indentfirst,epsfig,multirow,amsthm,color,array,amsfonts,graphicx,fullpage,hyperref,bm}
\usepackage[authoryear]{natbib}
\usepackage{threeparttable}
\usepackage{amssymb,amsthm,epsfig, upgreek}
\usepackage{bm}
\usepackage{multirow}
\usepackage[latin1]{inputenc}
\usepackage[figuresright]{rotating}
\usepackage{url,booktabs}
\newtheorem{TEO}{TEO}
\newtheorem{CEO}{CEO}

\newtheorem{Corollary}[CEO]{Corollary}
\newtheorem{proposition}[TEO]{Proposition}
\newtheorem{Remark}{Remark}
\usepackage{subfigure}

\journal{}
\begin{document}

\begin{frontmatter}


\title{A new class of fatigue life distributions}

\author[label4]{Marcelo Bourguignon\fnref{label1}}
 \address[label4]{Universidade Federal de Pernambuco\\
 Departamento de Estat\'istica, Cidade Universit\'aria, 50740-540 Recife, PE, Brazil
 \vspace{.1cm}}

\author[label4]{Rodrigo B. Silva\fnref{label2}}

\author[label4]{Gauss M. Cordeiro\fnref{label3}}

\fntext[label1]{Corresponding author.  Email: \texttt{m.p.bourguignon@gmail.com}}
\fntext[label2]{Email: \texttt{rodrigobs29@gmail.com}}
\fntext[label3]{Email: \texttt{gauss@de.ufpe.br}}

\begin{abstract}

In this paper, we introduce the Birnbaum-Saunders power series class of distributions which is
obtained by compounding Birnbaum-Saunders and power series distributions. The new class of distributions has as a particular case the two-parameter Birnbaum-Saunders distribution. The hazard rate function of the proposed class can be increasing and upside-down bathtub shaped.
We provide important mathematical properties such as moments, order statistics, estimation of the parameters and inference for large sample. Three special cases of the new class are investigated with some details. We illustrate the usefulness of the new distributions by means of two applications to real data sets. 
\end{abstract}

\end{frontmatter}

\section{Introduction}

Birnbaum and Saunders (1969) introduced a two-parameter distribution in order to model the failure time distribution for fatigue failure caused under cyclic loading, called the Birnbaum-Saunders distribution ($\mathcal{BS}$). Since then, this distribution has been extensively used for modelling failure times of fatiguing materials in several fields such as environmental and medical sciences, engineering, de\-mo\-gra\-phy, biolo\-gical studies, actuarial, economics, finance and insurance. A random variable $T$ following the $\mathcal{BS}(\alpha, \beta)$ distribution with shape parameter $\alpha > 0$ and scale parameter $\beta > 0$ is defined by
\begin{equation*}
 T = \beta \left\{\frac{\alpha Z}{2}+ \left[\left(\frac{\alpha Z}{2}\right)^2+1\right]^{1/2}\right\}^2,
\end{equation*}
where $Z$ is a standard normal random variable. The cumulative distribution function (cdf) of $T \sim \mathcal{BS}(\alpha, \beta)$ has the form
\begin{equation}\label{cdfbs}
 F_{\mathcal{BS}}(t; \alpha, \beta) = \Phi(\upsilon), \quad t>0,
\end{equation}
where $\upsilon = \alpha^{-1}\rho(t/\beta), \, \rho(z) = z^{1/2}-z^{-1/2}$ and $\Phi(\cdot)$ is the standard normal cumulative function. The corresponding probability density function (pdf) is
\begin{equation}\label{pdfbs}
 f_{\mathcal{BS}}(t; \alpha, \beta) = \kappa(\alpha, \beta)\,t^{-3/2}\,(t+\beta)\,\exp\left[-\frac{\tau(t/\beta)}{2\alpha^2}\right], \quad t>0,
\end{equation}
where $\kappa(\alpha, \beta) = \exp(\alpha^{-2})/ (2\alpha\sqrt{2\pi \beta})$ and $\tau(z)= z+z^{-1}$. The parameter $\beta$ is the median of the distribution, i.e. $F_{\mathcal{BS}}(\beta) = \Phi(0) = 1/2$. For any $k>0,\, kT \sim \mathcal{BS}(\alpha, k\beta)$. Additionally, the reciprocal property holds, i.e., $T^{-1} \sim \mathcal{BS}(\alpha, \beta^{-1})$, see Saunders (1974). The expected value and variance are $\operatorname{E}(T) = \beta (1 + \alpha^2/2)$ and $\operatorname{Var}(T) = (\alpha \beta)^2(1+5\alpha^2/4)$, respectively.

The $\mathcal{BS}$ distribution has a close relation to the normal distribution and has been widely studied in the statistical literature. For example, Ng \emph{et al}. (2006) proposed a simple bias-reduction method to reduce the biases of the maximum likelihood estimators (MLEs) for the $\mathcal{BS}$ model. They also discussed a Monte Carlo EM-algorithm for computing theses estimators. Kundu \emph{et al}. (2008) investigated the shape of the $\mathcal{BS}$ hazard function.

In the last few years, several classes of distributions were proposed by compounding some useful lifetime and power series distributions. Chahkandi and Ganjali (2009) studied the exponential power series ($\mathcal{EPS}$) family of distributions, which contains as special cases the exponential geometric ($\mathcal{EG}$), exponential Poisson ($\mathcal{EP}$) and exponential logarithmic ($\mathcal{EL}$) distributions, introduced by Adamidis and Loukas (1998), Kus (2007) and Tahmasbi and Rezaei (2008), respectively.  Morais and Barreto-Souza (2011) defined the Weibull power series ($\mathcal{WPS}$) class of distributions which includes as sub-models the $\mathcal{EPS}$ distributions.  The $\mathcal{WPS}$ distributions can have an increasing, decreasing and upside down bathtub failure rate function. The generalized exponential power series ($\mathcal{GEPS}$) distributions were proposed by Mahmoudi and Jafari (2012) in a similar way to that found in Morais and Barreto-Souza (2011). In a very recent paper, Silva \
emph{et al}. (2012) introduced the extended Weibull power series ($\mathcal{EWPS}$) class of distributions which includes as special models the $\mathcal{EPS}$ and $\mathcal{WPS}$ distributions.

In a similar manner, we introduce the \emph{Birnbaum-Saunders power series} ($\mathcal{BSPS}$) class of distributions, which is defined by compounding the Birnbaum-Saunders and power series distributions. The main aim of this paper is to introduce a new class of distributions, which extends the Birnbaum-Saunders distribution, with the hope that the new distribution may have a 'better fit' in certain practical situations. Additionally, we will provide a comprehensive account of the mathematical properties of the proposed new class of distributions. The hazard rate function of the proposed class can be increasing and upside-down bathtub shaped. We are motivated to introduce the $\mathcal{BSPS}$ distributions because of the wide usage of~\eqref{cdfbs} and the fact that the current generalization provides means of its continuous extension to still more complex situations.

This paper is unfolds as follows. In Section 2, a new class of distributions is obtained by mixing the $\mathcal{BS}$ and zero truncated power series distributions. The mixing procedure was previously carried out by Morais and Barreto-Souza (2011) and Silva \emph{et al}. (2012). In Section 3, some properties of the new class are discussed. In Section 4, the estimation of the model parameters is performed by the method of maximum likelihood. In Section 5, we introduce and study three special cases of the $\mathcal{BSPS}$ class of distributions. In Section 6, two illustrative applications based on real data are provided. Finally, concluding remarks are presented in Section 7.

\section{The new class}
Let $T_1, \ldots, T_N$ be independent and identically distributed (iid) $\mathcal{BS}$ random variables with cdf~\eqref{cdfbs} and pdf~\eqref{pdfbs}. We assume that the random variable $N$ has a zero truncated power series distribution with probability mass function
\begin{equation}\label{powerseries}
p_n = P(N=n)=\frac{a_n \, \theta^n}{C(\theta)},\,n=1,2,\ldots,
\end{equation}
where $a_n$ depends only on $n$, $C(\theta) = \sum_{n=1}^{\infty}a_n \,\theta^n$ and $\theta>0$ is such that $C(\theta)$ is finite. Table \ref{table1}, reported in Morais and Barreto-Souza (2011), summarizes some power series distributions (truncated at zero) defined according to (\ref{powerseries}) such as the Poisson, logarithmic, geometric and binomial distributions.

By assuming that the random variables $N$ and $T$ are independent, we define $X=\min(T_1, \ldots, T_N)$. Then,
\begin{equation*}
P(X\leq x, N=n) = \frac{a_n\, \theta^n}{C(\theta)}\left[1-(1-\Phi(\upsilon))^n\right], \quad x>0,
\quad n \geq 1.
\end{equation*}

The $\mathcal{BSPS}$ class of distributions is defined by the marginal cdf of $X$:
\begin{equation}\label{cdf}
F_{\mathcal{BSPS}}(x; \theta, \alpha, \beta) = 1-\frac{C\left[\theta(1-\Phi(\upsilon))\right]}{C(\theta)}, \quad x>0.
\end{equation}

\begin{table}[!htbp]
\centering
\scalebox{0.85}[0.85]{
\begin{tabular}{llllllllllll}
 \toprule
Distribution && $a_n$&& $C(\theta)$ && $C'(\theta)$&&$C''(\theta)$&&$C(\theta)^{-1}$&$\Theta$\\
 \midrule
Poisson&&$n!^{-1}$ && $\mathrm{e}^{\theta}-1$&&$e^{\theta}$&&$e^{\theta}$&&$\log(\theta+1)$&$\theta \in (0, \infty)$\\
&&&&&\\
Logarithmic&&$n^{-1}$&&$-\log(1-\theta)$&&$(1-\theta)^{-1}$&&$(1-\theta)^{-2}$&&$1-\mathrm{e}^{-\theta}$&$\theta \in (0,1)$\\
&&&&&\\
Geometric&&1&&$\theta(1-\theta)^{-1}$&&$(1-\theta)^{-2}$&&$2(1-\theta)^{-3}$&&$\theta(\theta+1)^{-1}$&$\theta \in (0,1)$\\
&&&&&\\
Binomial&&$\binom{m}{n}$&&$(\theta+1)^{m} -
1$&&$m(\theta+1)^{m-1}$&&$\dfrac{m(m-1)}{(\theta+1)^{2-m}}$&&$(\theta-1)^{1/m}-1$&$\theta \in (0, 1)$\\
\bottomrule
\end{tabular}}
\caption{Useful quantities for some power series distributions.}\label{table1}
\end{table}

Hereafter, the random variable $X$ following~\eqref{cdf} with parameters $\theta, \alpha$ and $\beta$ is denoted by $X\sim \mathcal{BSPS}(\theta,\alpha, \beta)$.
The pdf to~\eqref{cdf} is
\begin{equation}\label{pdf}
 f_{\mathcal{BSPS}}(x; \theta, \alpha, \beta) = \theta f_{\mathcal{BS}}(x) \frac{C'\left[\theta(1-\Phi(\upsilon))\right]}{C(\theta)}, \quad x>0,
\end{equation}
where $f_{\mathcal{BS}}(\cdot)$ is given in (\ref{pdfbs}).

\begin{proposition}
The classical $\mathcal{BS}$ distribution with parameters $\alpha$ and $\beta$ is a limiting special case of the proposed model when $\theta \rightarrow 0$.
\end{proposition}
\begin{proof}
For $x>0$, we have
\begin{align*}
\lim_{\theta \rightarrow 0^+} F_{\mathcal{BSPS}}(x) &= 1 - \lim_{\theta \rightarrow 0^+} \frac{\displaystyle{\sum_{n=1}^\infty a_n\left[\theta(1-\Phi(\upsilon))\right]^n}}{\displaystyle{\sum_{n=1}^\infty a_n\, \theta^n}}\\
&= 1 - \lim_{\theta \rightarrow 0^+} \frac{\displaystyle{1-\Phi(\upsilon) + a_1^{-1} \sum_{n=2}^\infty a_n \,\theta^{n-1}(1-\Phi(\upsilon))^n}}{\displaystyle{1 + a_1^{-1}\sum_{n=2}^\infty a_n\, \theta^{n-1}}}=\Phi(\upsilon).
\end{align*}
\end{proof}

We now derive two useful expansions for the density function~\eqref{pdf}. We have $C'(\theta) = \sum_{n=1}^{\infty}n\, a_n\,\theta^{n-1}$ and then
\begin{eqnarray}\label{expansion1}
 f_{\mathcal{BSPS}}(t; \theta, \alpha, \beta) &=& f_{\mathcal{BS}}(t; \alpha, \beta) \sum_{n=1}^{\infty}n \frac{a_n\,\theta^{n}}{C(\theta)}[1 - \Phi(\upsilon)]^{n-1}.
\end{eqnarray}

By using the binomial theorem in equation (\ref{expansion1}), we can write
\begin{eqnarray}\label{expansion2}
 f_{\mathcal{BSPS}}(x; \theta, \alpha, \beta) &=& f_{\mathcal{BS}}(x; \alpha, \beta) \sum_{n=1}^{\infty} \sum_{k=0}^{n-1} \omega_{n,k}\Phi(\upsilon)^k,
\end{eqnarray}
where $\omega_{n,k} = (-1)^k\,n\,p_n  \dbinom{n-1}{k}$.

\begin{proposition} The density function of $X$ can be expressed as an infinite linear combination of densities of minimum order statistics of $T$. \end{proposition}
\begin{proof}
We know that $C'(\theta) = \sum_{n=1}^{\infty}n\, a_n\,\theta^{n-1}$. Therefore,
\begin{equation*}\label{ordrpdf}
f_{\mathcal{BSPS}}(x; \theta, \alpha, \beta) = \sum\limits_{n=1}^{\infty}p_n\,f_{T_{(1)}}(x; \alpha,\beta, n), \quad x>0,
\end{equation*}
where $f_{T_{(1)}}(x; \alpha,\beta, n)$ is the density function of $T_{(1)}=\min(T_1, \ldots, T_n)$, for fixed $n$, given by
$$f_{T_{(1)}}(t;\alpha,\beta,n) = n\,f_{\mathcal{BS}}(t)[1-\Phi(v)]^{n-1}, \quad t>0.$$
\end{proof}

\begin{Remark} Let $Y = \max(T_1, \ldots, T_N)$, then the cdf and pdf of $Y$ are
\begin{equation}\label{CBSPS}
F_{Y}(y; \theta, \alpha,\beta) = \frac{C(\theta\,\Phi(v))}{C(\theta)}, \quad y>0
\end{equation}
and
$$f_{Y}(y; \theta, \alpha,\beta) = \theta f_{\mathcal{BS}}(y)\frac{C'(\theta\,\Phi(v))}{C(\theta)}, \quad y>0.$$
\end{Remark}

Note that
$$f_{Y}(y; \theta, \alpha,\beta) = \sum\limits_{n=1}^{\infty}p_n\,f_{T_{(n)}}(y; \alpha,\beta,n), \quad y>0,$$
where $f_{T_{(n)}}(y;\alpha,\beta,n)$ is the density function of $T_{(n)}=\max(T_1, \ldots, T_n)$, for fixed $n$, given by
$$f_{T_{(n)}}(t;\alpha,\beta,n) = n\,f_{\mathcal{BS}}(t)\,\Phi(v)^{n-1}, \quad t>0.$$

The distribution with cdf (\ref{CBSPS}) is called the \emph{complementary Birnbaum-Saunders power series} ($\mathcal{CBSPS}$) distribution. This class of distributions is a suitable model in a complementary risk problem based in the presence of latent risks which arise in several areas such as public health, actuarial science, biomedical studies, demography and industrial reliability (Basu and Klein, 1982). However, in this work, we do not focus on this alternative class of distributions.

The $\mathcal{BSPS}$ survival function becomes
\begin{equation*}
S_{\mathcal{BSPS}}(x; \theta, \alpha, \beta) = \frac{C\left[\theta(1-\Phi(\upsilon))\right]}{C(\theta)}, \quad x>0
\end{equation*}
and the corresponding hazard rate function reduces to
\begin{equation*}
h_{\mathcal{BSPS}}(x; \theta, \alpha, \beta) = \theta f_{\mathcal{BS}}(x) \frac{C'\left[\theta(1-\Phi(\upsilon))\right]}{C\left[\theta(1-\Phi(\upsilon))\right]}, \quad x>0.
\end{equation*}

\section{Quantiles, order statistics and moments}

The quantile function of $X$, say $Q(u)$, is given by
\begin{equation*}
 t = Q(u) = \beta\left\{\frac{\alpha\, \delta}{2} + \left[\left(\frac{\alpha\, \delta}{2}\right)^2 + 1\right]^{1/2}\right\}^2,
\end{equation*}
where $\delta = -\Phi^{-1}\left\{C^{-1}\left[(1-u)C(\theta)\right]/\theta\right\}$, $\Phi^{-1}(\cdot)$ is the inverse cumulative function of the standard normal distribution, $C^{-1}(\cdot)$ is the inverse function of $C(\cdot)$ and $u$ is a uniform random number on the unit interval $(0,1)$.

Let $X_1, \ldots, X_m$ be a random sample with density function~\eqref{pdf} and define the $i$th order statistic by $T_{i:m}$. The density function of $T_{i:m}$, say $f_{i:m}(x; \theta, \alpha, \beta)$, is given by
\begin{eqnarray*}\label{ordstat}\nonumber
 f_{i:m}(x; \theta, \alpha, \beta) &=& \frac{1}{B(i, m-i+1)}f_{\mathcal{BSPS}}(x; \theta, \alpha, \beta)F_{\mathcal{BSPS}}(x; \theta, \alpha, \beta)^{i-1}\left[1-F_{\mathcal{BSPS}}(x; \theta, \alpha, \beta)\right]^{m-i}\\
 &=& \frac{f_{\mathcal{BS}}(x; \alpha, \beta)}{B(i, m-i+1)} \sum_{k=0}^{i-1}(-1)^k \binom{i-1}{k}S_{\mathcal{BSPS}}(x; \theta, \alpha, \beta)^{m+k-i},
 \end{eqnarray*}
where $S_{\mathcal{BSPS}}(\cdot)$ is the $\mathcal{BSPS}$ survival function.

\begin{proposition}
The $\mathcal{BSPS}$ order statistics can be expressed as a linear combination of exponentiated Birnbaum-Saunders ($\mathcal{EBS}$) distributions with parameters $\alpha, \beta$ and $s+1$.
\end{proposition}
In fact,
\begin{equation*}
f_{i:m}(x; \theta, \alpha, \beta) = \sum_{k=0}^{i-1}\sum_{j=0}^{\infty}\sum_{s=0}^{m+k-i+j}\omega_{k,j,s}\, h_{s+1}(x, \alpha, \beta),
\end{equation*}
where $$\omega_{k,j,s} = \frac{(-1)^{k+s}\binom{i-1}{k}\binom{m+k+j-i}{s} d_j}{(s+1)B(i, m-i+1)C(\theta)^{m+k-i}},$$ $$d_0=a_1^{m+k-i}, \quad d_r = \frac{1}{r d_0} \sum_{j=1}^{r}\left[j(m+k-i)-r+j\right]a_{j+1} d_{r-j}$$ and $h_{s+1}(t, \alpha, \beta)$ denotes the $\mathcal{EBS}$ density function with parameters $\alpha, \beta$ and $s+1$.

Many of the important characteristics and features of a distribution are obtained through the moments. The ordinary moments of $X$ can be derived from the probability weighted moments (PWMs) (Greenwood \emph{et al}., 1979) of the $\mathcal{BS}$ distribution formally defined for $p$ and $r$ non-negative integers by
\begin{equation}\label{tau}
 \tau_{p,r} = \kappa(\alpha, \beta)\int_{0}^\infty t^{p-3/2} (t+\beta) \exp\left[-\frac{\tau(t/\beta)}{2\alpha^2}\right] \Phi^r(\upsilon) \mathrm{d}t.
\end{equation}

The integral (\ref{tau}) can be easily computed numerically in software such as MAPLE, MATHEMATICA, Ox and R, see Cordeiro and Lemonte (2011). An alternative expression for $\tau_{p,r}$ can be derived from Cordeiro and Lemonte (2011) as
\begin{align*}\label{pwmbs}
\begin{split}
\tau_{p,r}&=\frac{\beta^p}{2^r}\sum_{j=0}^r \binom{r}{j} \sum_{k_1,\ldots,k_j=0}^\infty
A(k_1,...,k_j)\sum_{m=0}^{2s_j+j}(-1)^m \binom{2s_j+j}{m}\\
&\qquad\times\beta^{-(2s_j+j-2m)/2}\,I\bigl(p+(2s_j+j-2m)/2,\alpha\bigr).
\end{split}
\end{align*}

Here, $s_j=k_1+\cdots+k_j$, $A(k_1,\ldots,k_j)=\alpha^{-2s_j-j}a_{k_1}\ldots a_{k_j}$, $a_k=(-1)^k2^{(1-2k)/2}\{\sqrt{\pi}(2k+1)k!\}^{-1}$ and
\begin{equation*}
I(p,\alpha)=\frac{K_{p+1/2}(\alpha^{-2})+K_{p-1/2}(\alpha^{-2})}{2K_{1/2}(\alpha^{-2})},
\end{equation*}
where the function $K_{\nu}(z)$ denotes the modified Bessel function of the third kind with $\nu$ representing its order and $z$ the argument.
Its integral representation is $K_{\nu}(z)=0.5 \int_{-\infty}^{\infty} \exp\{-z\cosh(t)-\nu t\}dt$.

Thus, the $s$th moment of $X$ can be obtained from equations (\ref{expansion2}) and (\ref{tau}) as
\begin{eqnarray*}
\textrm{E}(X^s) &=& \sum_{n=1}^\infty \sum_{k=0}^{n-1} \omega_{n,k} \, \tau_{s,k}.
\end{eqnarray*}

Expressions for the $s$th moments of the order statistics $X_{1:m}, \ldots, X_{m,m}$ with cumulative function~\eqref{cdf} can be obtained using a result due to Barakat and Abdelkader (2004). We have
\begin{align*}\label{omordstat}\nonumber
 \mathrm{E}(X^{s}_{i:m}) &= s \sum_{j=m-i+1}^m (-1)^{j-m+i-1}\binom{j-1}{m-i}\binom{m}{j}\int_{0}^{\infty}x^{s-1}S_{\mathcal{BSPS}}(x, \theta, \alpha, \beta)^{j}\mathrm{d}x\\
 &= s \sum_{j=m-i+1}^m \frac{(-1)^{j-m+i-1}}{C(\theta)^j}\binom{j-1}{m-i}\binom{m}{j}\int_{0}^{\infty}x^{s-1} C\left[\theta(1-\Phi(\upsilon))\right]^j \mathrm{d}x,
\end{align*}
for $i = 1, \ldots, m$.

\section{Maximum likelihood estimation}\label{mv}
Here, we discuss maximum likelihood estimation and inference for the $\mathcal{BSPS}$ distribution. Let $x_1, \ldots, x_n$ be a random sample from $X$ and let $\boldsymbol{\Theta} =(\theta, \alpha, \beta)^\top$ be the vector of the model parameters. The log-likelihood function for
$\boldsymbol{\Theta}$ reduces to
\begin{eqnarray*}\label{vero}
\ell(\theta, \alpha, \beta) &=& n\{\log(\theta) - \log[C(\theta)] + \log[\kappa(\alpha, \beta)]\} - \frac{3}{2}\sum\limits_{i=1}^{n}\log(x_i) + \sum\limits_{i=1}^{n}\log(x_i + \beta)\\
&-& \frac{1}{2\alpha^2}\sum\limits_{i=1}^{n}\tau\left(\frac{x_i}{\beta}\right) + \sum\limits_{i=1}^{n}\log\{C'\left[\theta(1-\Phi(v_i))\right]\}.\\
\end{eqnarray*}

The components of the score vector $U(\boldsymbol{\Theta})$ are given by

\begin{align*}
U_{\theta}(\boldsymbol{\Theta}) &= \frac{n}{\theta} - n\frac{C'(\theta)}{C(\theta)} + \sum_{i=1}^n \frac{1-\phi(\upsilon_i)}{C'\left[\theta(1-\Phi(\upsilon_i))\right]}\frac{\partial C'\left[\theta(1-\Phi(\upsilon_i))\right]}{\partial \theta},\\
U_{\alpha}(\boldsymbol{\Theta}) &= \frac{n}{\alpha}\left(1+\frac{2}{\alpha^2}\right) + \frac{1}{\alpha^3} \sum_{i=1}^n \tau \left(\frac{x_i}{\beta}\right) + \frac{\theta}{\alpha^3} \sum_{i=1}^n \frac{\phi(\upsilon_i)\rho(x_i/\beta)}{C'\left[\theta(1-\Phi(\upsilon_i))\right]}\frac{\partial C'\left[\theta(1-\Phi(\upsilon_i))\right]}{\partial \alpha} \\
\intertext{and}
U_{\beta}(\boldsymbol{\Theta})  &= -\frac{n}{2\beta} + \frac{1}{2\alpha^2\beta}\sum_{i=1}^n \left(\frac{x_i}{\beta}+\frac{\beta}{x_i}\right) - \frac{\theta}{2\alpha \beta}\sum_{i=1}^n \frac{\phi(\upsilon_i)\tau(\sqrt{x_i/\beta})}{C'\left[\theta(1-\Phi(\upsilon_i))\right]}\frac{\partial C'\left[\theta(1-\Phi(\upsilon_i))\right]}{\partial \beta},
\end{align*}
where $\phi(\cdot)$ is the standard normal density and $v_i = \alpha^{-1} \rho(x_i/\theta)$.

Setting these equations to zero, $U(\boldsymbol{\Theta})=0$, and solving them simultaneously yields the MLE $\widehat{\boldsymbol{\Theta}}$ of $\Theta$. These equations cannot be solved analytically and statistical software can be used to compute them numerically using iterative techniques such as the Newton-Raphson algorithm.

Often with lifetime data and reliability studies, one encounters censoring. A very simple random censoring mechanism that is often realistic is one in which each individual $i$ is assumed to have a lifetime $X_i$ and a censoring time $C_i$, where $X_i$ and $C_i$ are independent
random variables. Suppose that the data consist of $n$ independent observations $x_i = \textrm{min}(X_i, C_i)$ and $\delta_i = I(X_i \leq C_i)$ is such that $\delta_i
= 1$ if $ X_i$ is a time to event and $\delta_i = 0$ if it is right censored for $i = 1, \ldots, n$. The censored likelihood $L(\theta, \alpha, \beta)$ for the model parameters is
\begin{equation*}
L(\theta, \alpha, \beta) \propto \prod_{i=1}^{n}\,[f_{\mathcal{BSPS}}(x_i;\theta, \alpha, \beta)]^{\delta_i}\,[S_{\mathcal{BSPS}}(x_i;\theta, \alpha, \beta)]^{1-\delta_i}.
\end{equation*}

In order to perform interval estimation and hypothesis tests on the model parameters $\alpha$, $\beta$ and $\theta$, the normal approximation for the MLE can be applied. Note that under certain conditions for the parameters, the asymptotic distribution of $\sqrt{n}(\widehat{\boldsymbol{\Theta}}-\boldsymbol{\Theta})$ is
multivariate normal $N_{3}(0,I^{-1}(\boldsymbol{\Theta}))$, where $J_n(\boldsymbol{\theta})$ is the observed information matrix and $I(\boldsymbol{\Theta})=\lim_{n\rightarrow\infty}J_n(\boldsymbol{\Theta})$. Note also that an $100(1-\gamma)\%$ ($0<\gamma<1/2$) asymptotic confidence interval for the $i$th parameter $\theta_i$ in $\Theta$
is specified by
$$ACI_i=(\widehat\theta_i-z_{1-\gamma/2}\sqrt{\widehat J^{\theta_i,\theta_i}},\hat{\theta_i} +z_{1-\gamma/2}\sqrt{\widehat J^{\theta_i,\theta_i}}),$$
where $\widehat J^{\theta_i,\theta_i}$ stands for the $i$th diagonal element of the inverse of the observed information matrix
estimated at $\widehat{\boldsymbol{\Theta}}$, i.e., $J_n(\widehat{\boldsymbol{\Theta}})^{-1}$,
for $i=1,2,3$, and $z_{1-\gamma/2}$ is the $1-\gamma/2$ standard normal quantile.

For testing the goodness-of-fit of the $\mathcal{BSPS}$ distribution and for comparing it with some sub-models, we can use the likelihood ratio (LR) statistics. The LR statistic for testing the null hypothesis $\mathcal{H}_0$: $\boldsymbol{\Theta}_1 = \boldsymbol{\Theta}_{1}^{(0)}$ versus the
alternative hypothesis $\mathcal{H}_1$: $\boldsymbol{\Theta}_1 \neq
\boldsymbol{\Theta}_{1}^{(0)}$ is given by $w =
2\{\ell(\widehat{\boldsymbol{\Theta}}) -
\ell(\widetilde{\boldsymbol{\Theta}})\}$, where
$\widehat{\boldsymbol{\Theta}}$ and $\widetilde{\boldsymbol{\Theta}}$ are
the MLEs under the alternative and null hypotheses, respectively.
Under the null hypothesis, $\mathcal{H}_0$, $w{\xrightarrow{d}}\chi_{k}^{2}$, where $k$ is the dimension
of the subset $\boldsymbol{\Theta}_{1}$ of interest.

\section{Special cases}

In this section, we investigate some special models of the $\mathcal{BSPS}$ class of distributions. We offer
some expressions for the moments and moments of the order statistics. We illustrate the flexibility of these distributions and provide plots of the density and hazard rate functions for selected parameter values.

\subsection{Birnbaum-Saunders geometric distribution}
\vskip3mm
The {\it Birnbaum-Saunders geometric} ($\mathcal{BSG}$) distribution is defined by the cdf (\ref{cdf}) with $C(\theta) = \theta(1-\theta)^{-1}$ leading to
\begin{equation}\label{pdfbsg}
F_{\mathcal{BSG}}(x;\theta, \alpha, \beta) = 1 - \frac{(1-\theta)[1 - \Phi(v)]}{1-\theta[1-\Phi(v)]}, \quad x > 0,
\end{equation}
where $\theta \in (0,1)$.

\begin{proposition} The distribution of the form (\ref{pdfbsg}) is geometric minimum stable.
\end{proposition}
\begin{proof} The proof follows easily of the arguments given by Marshall and Olkin (1997, p. 647). We omit the details.
\end{proof}

The associated density and hazard rate functions are
\begin{equation*}
f_{\mathcal{BSG}}(x;\theta, \alpha, \beta)  = (1 - \theta)f_{\mathcal{BS}}(x)\left[1 - \theta(1-\Phi(v))\right]^{-2}
\end{equation*}
and
\begin{equation}\label{hBSG}
h_{\mathcal{BSG}}(x;\theta, \alpha, \beta) = \frac{f_{\mathcal{BS}}(x)}{[1-\Phi(v)]\{1 - \theta[1-\Phi(v)]\}} = \frac{h_{\mathcal{BS}}(x)}{1 - \theta[1-\Phi(v)]},
\end{equation}
for $x > 0$, respectively. From equation (\ref{hBSG}), we note that $h_{\mathcal{BSG}}(x;\theta, \alpha, \beta)/h_{\mathcal{BS}}(x;\alpha, \beta)$ is decreasing in $x$.
It can also be shown that $$\lim\limits_{x \rightarrow 0} h_{\mathcal{BSG}}(x;\theta, \alpha, \beta) = 0 \,\,\,\,\, \mathrm{and} \,\,\,\,\, \lim\limits_{x \rightarrow \infty} h_{\mathcal{BSG}}(x;\theta, \alpha, \beta) = \frac{1}{2\alpha^2\beta}.$$

\begin{Corollary} Let $X \sim \mathcal{BSG}(\theta,\alpha,\beta)$ and $T \sim\mathcal{BS}(\alpha,\beta)$. Then $h_{\mathcal{BSG}}(x;\theta,\alpha,\beta) \geq h_{\mathcal{BS}}(t;\alpha,\beta)$ for all $\theta \in (0,1)$.
\end{Corollary}
\begin{proof} It is straightforward.
\end{proof}

Cancho \emph{et al}. (2012) proposed and studied the geometric Birnbaum-Saunders regression with cure rate. However, they considered $C(\theta) = (1 - \theta)^{-1}$ instead of $C(\theta) = \theta (1 - \theta)^{-1}$. In Figure \ref{figBSG}, we plot the density and hazard rate functions of the $\mathcal{BSG}$ distribution for selected parameter values.  We can verify that this distribution has an upside-down bathtub or an increasing failure rate function depending on the values of its parameters.

\begin{figure}[!htbp]
\centering
\subfigure[$\alpha = 0.2$ and $\beta = 5.0$]{\includegraphics[scale=0.5]{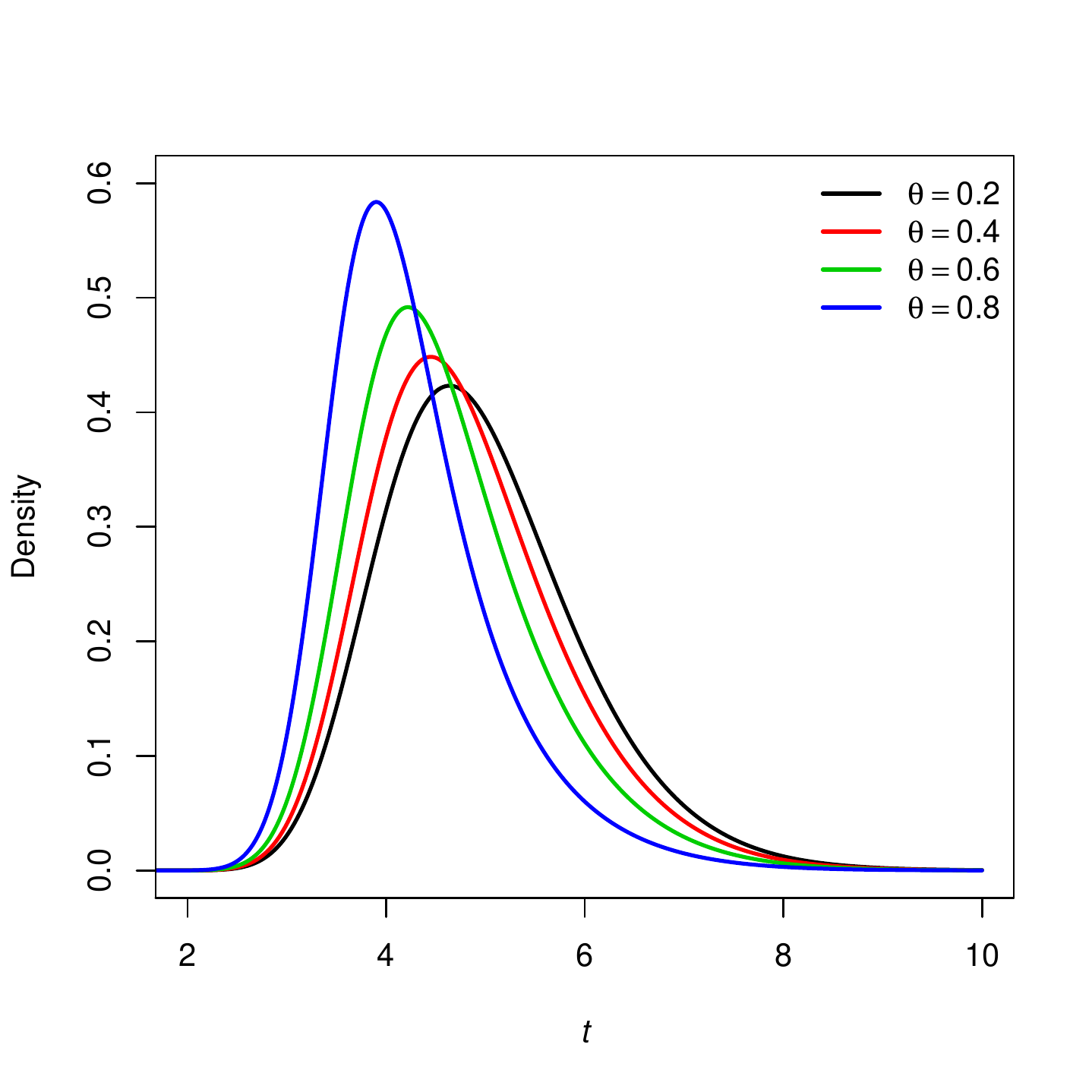}}
\subfigure[$\alpha = 2.0$ and $\beta = 5.0$]{\includegraphics[scale=0.5]{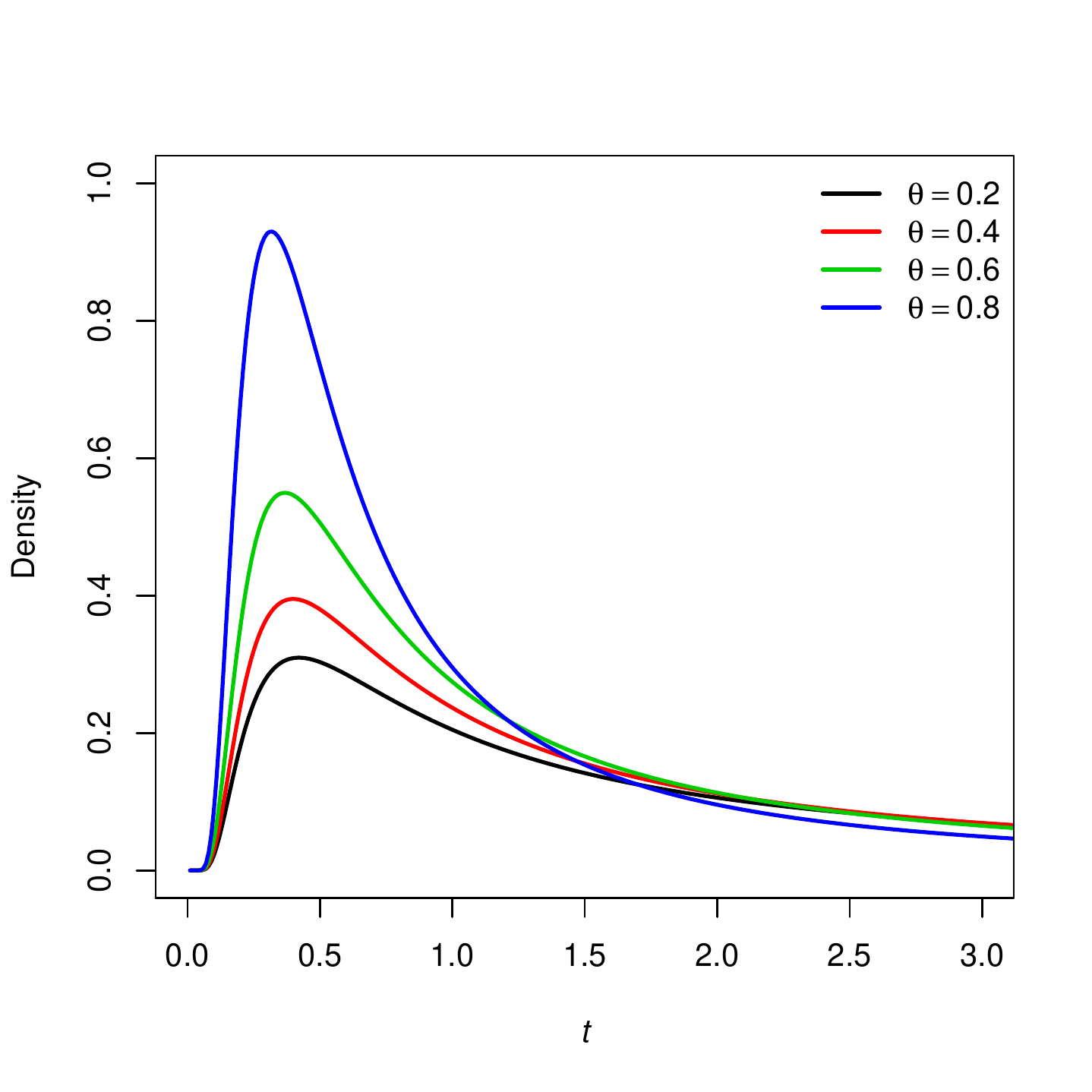}}
\subfigure[$\alpha = 0.2$ and $\beta = 5.0$]{\includegraphics[scale=0.5]{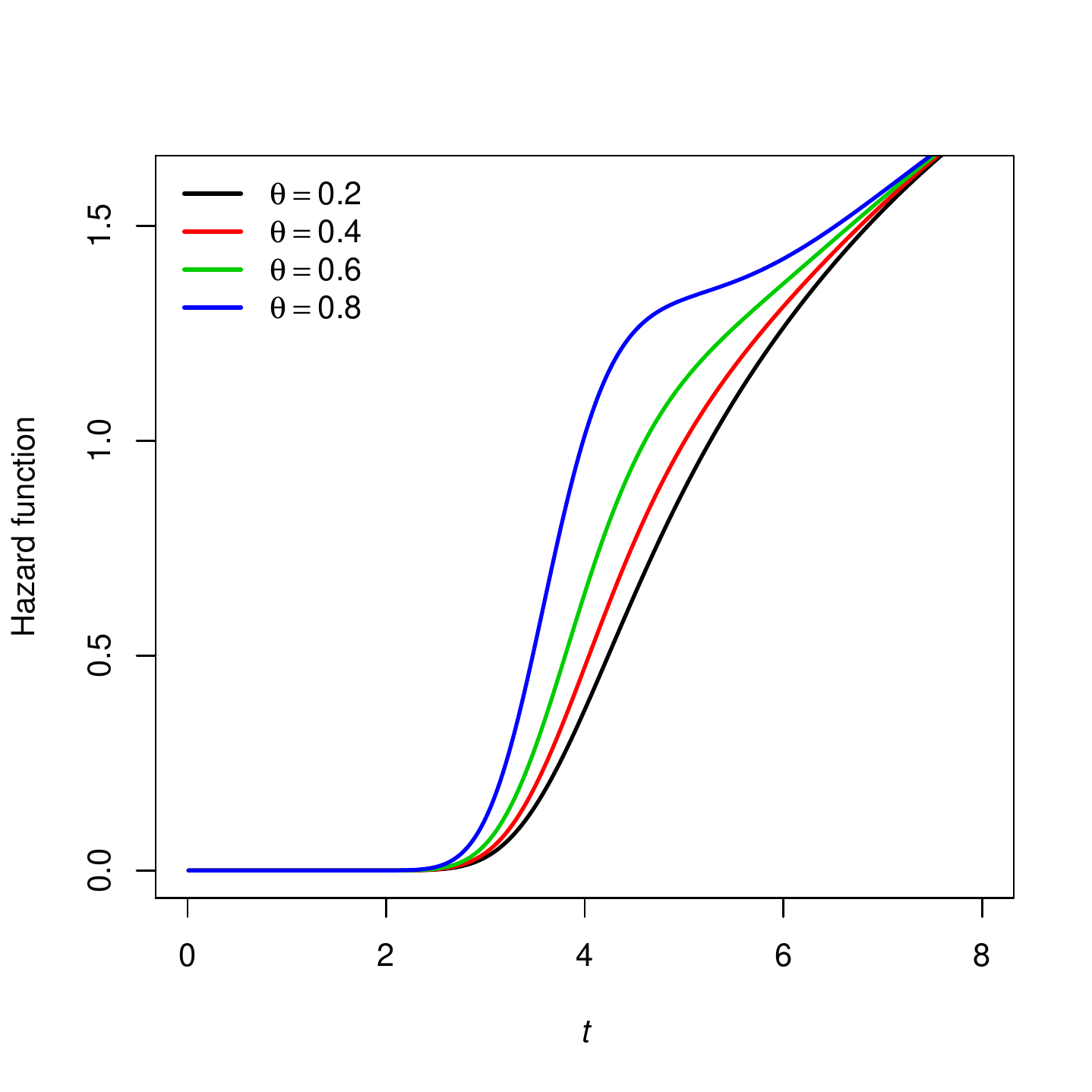}}
\subfigure[$\alpha = 2.0$ and $\beta = 5.0$]{\includegraphics[scale=0.5]{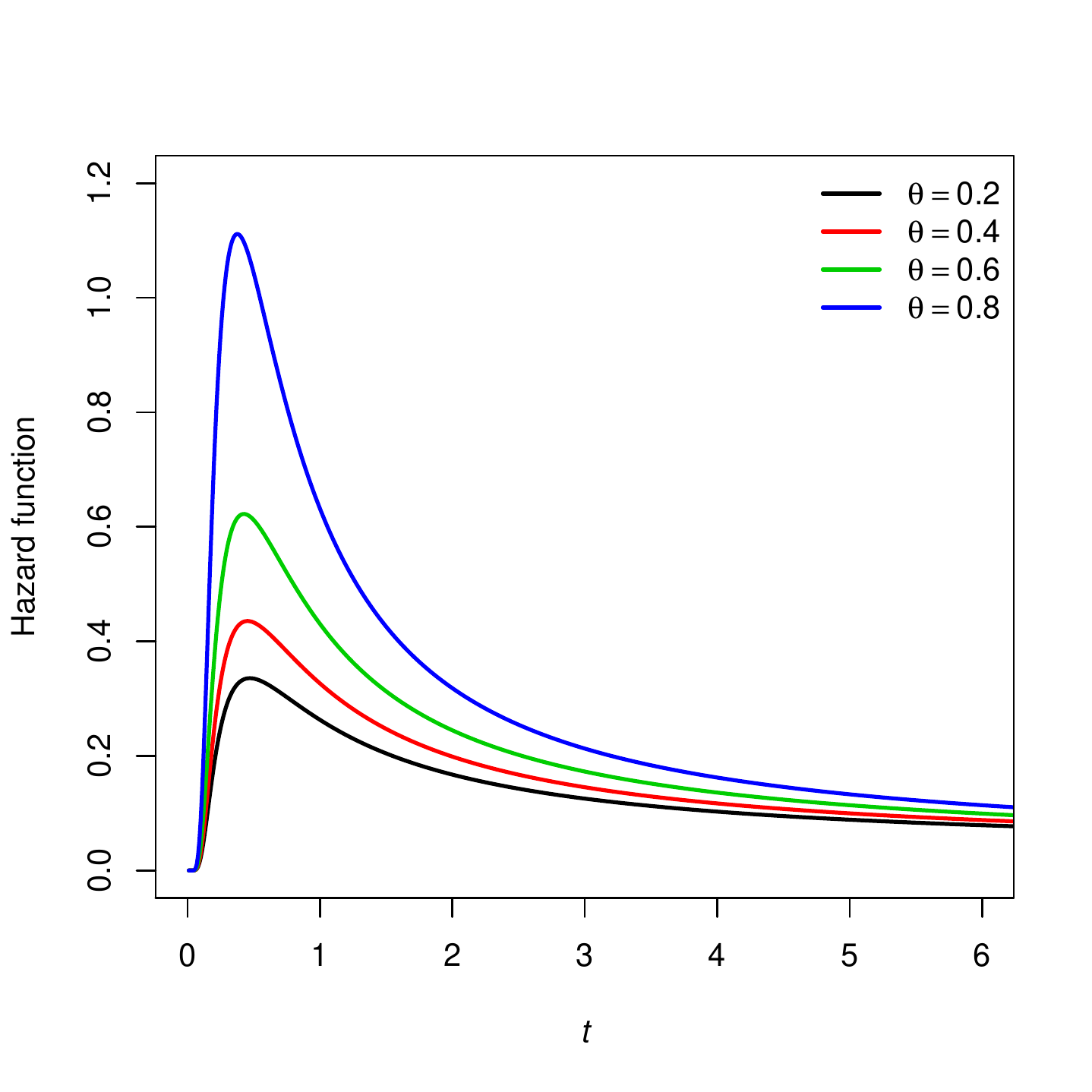}}
\caption{Plots of the $\mathcal{BSG}$ density and hazard rate functions for some parameter values.}
\label{figBSG}
\end{figure}

Expressions for the density function and the moments of the order statistics $X_{1:m}, \ldots, X_{m:m}$ from a random sample of the $\mathcal{BSG}$ distribution are given by
\begin{equation*}
f_{i:m}(x) = \frac{f_{\mathcal{BS}}(x; \alpha, \beta)}{B(i, m-i+1)} \sum_{k=0}^{i-1} (-1)^k \binom{i-1}{k} \left[\frac{(1-\theta)\Phi(-\upsilon)}{1-\theta\Phi(-\upsilon)}\right]^{m+k-i}
\end{equation*}
and
\begin{align*}
 \mathrm{E}(X^{s}_{i:m}) &= s \sum_{j=m-i+1}^m \frac{(-1)^{j-m+i-1}}{\theta^j(1-\theta)^{-j}}\binom{j-1}{m-i}\binom{m}{j}\int_{0}^{\infty}x^{s-1} \left[1-\theta\Phi(-\upsilon)\right]^j \mathrm{d}x,
\end{align*}
which are easily determined numerically.

Let $x_1, \cdots, x_n$ be a random sample of size $n$ from $X\sim \mathcal{BSG}(\delta,\alpha,\beta)$. The log-likelihood function for the vector of parameters $\boldsymbol{\Theta} = (\theta,\alpha, \beta)^\top$ can be expressed as
\begin{eqnarray*}
\ell(\boldsymbol{\Theta}) &=& n\{\log(1-\theta) + \log[\kappa(\alpha, \beta)]\} - \frac{3}{2}\sum\limits_{i=1}^{n}\log(x_i) + \sum\limits_{i=1}^{n}\log(x_i + \beta)\\
&-& \frac{1}{2\alpha^2}\sum\limits_{i=1}^{n}\tau\left(\frac{x_i}{\beta}\right) - 2\sum\limits_{i=1}^{n}\log[1 - \theta(1-\Phi(v_i))].
\end{eqnarray*}

\subsection{Birnbaum-Saunders Poisson distribution}

The \textit{Birnbaum-Saunders Poisson} ($\mathcal{BSP}$) distribution follows by taking $C(\theta) = \mathrm{e}^{\theta}-1$ in~\eqref{cdf}, which yields

\begin{equation*}
F_{\mathcal{BSP}}(t;\theta, \alpha, \beta) = 1 - \frac{\exp\{\theta(1-\Phi(v))\}-1}{\exp(\theta)-1}.
\end{equation*}

The density and hazard rate functions of the $\mathcal{BSP}$ distribution are
\begin{equation*}
f_{\mathcal{BSP}}(t;\theta, \alpha, \beta) = \frac{\theta f_{\mathcal{BS}}(t)\exp[\theta(1-\Phi(v))]}{\exp(\theta)-1}
\end{equation*}
and
\begin{equation*}
h_{\mathcal{BSP}}(t;\theta, \alpha, \beta) = \frac{\theta f_{\mathcal{BS}}(t)\exp\{\theta(1-\Phi(v))\}}{\exp\{\theta(1-\Phi(v))\}-1}.
\end{equation*}

\begin{Remark} The limit of the $\mathcal{BSP}$ hazard rate function as $t \rightarrow 0$ is $0$.
\end{Remark}

In Figure \ref{figBSP}, we plot the density and hazard rate functions of the $\mathcal{BSP}$ distribution for selected parameter values. This distribution can have an upside-down bathtub or an increasing failure rate function depending on the values of its parameters.

\begin{figure}[!htbp]
\centering
\subfigure[$\alpha = 0.2$ and $\beta = 1.0$]{\includegraphics[scale=0.5]{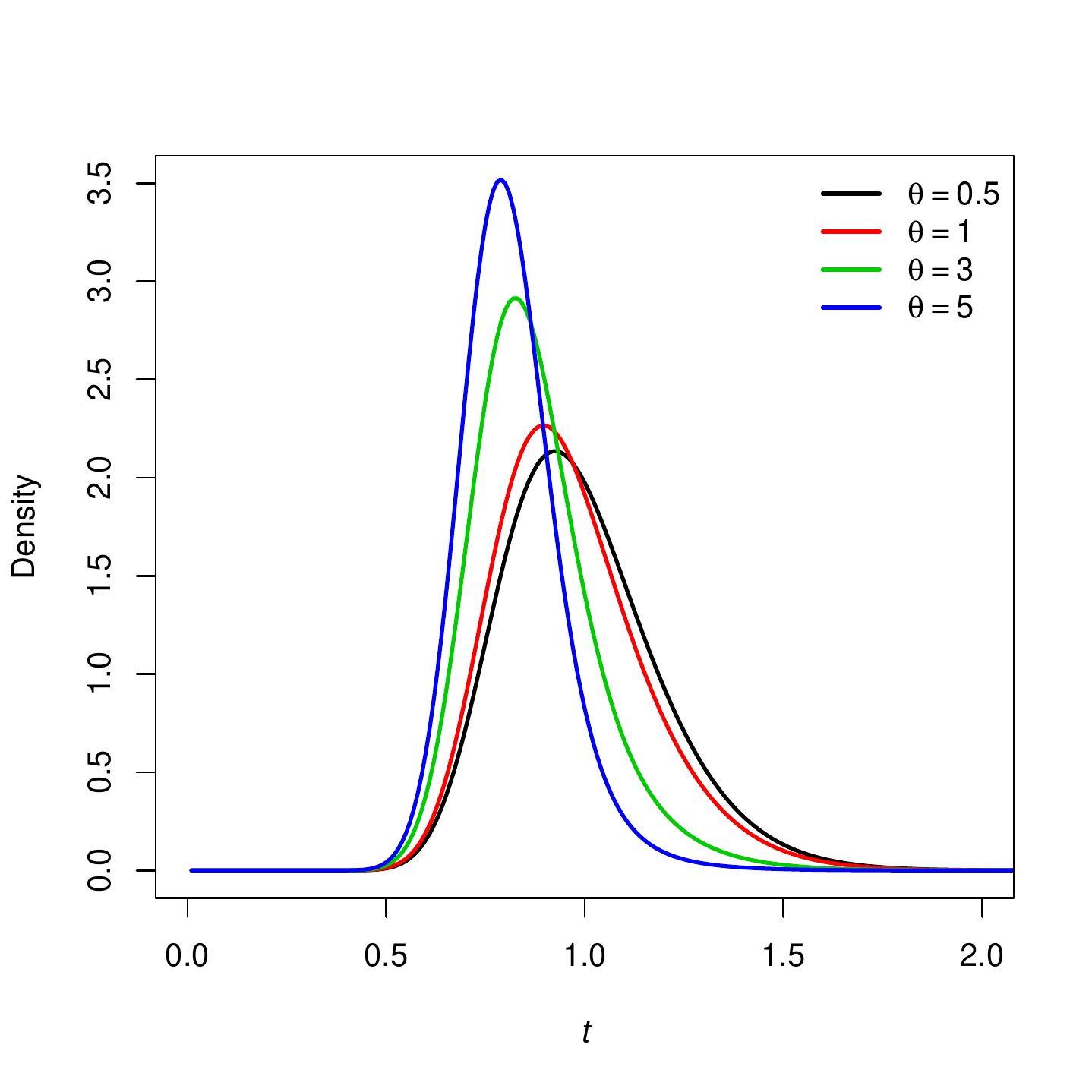}}
\subfigure[$\alpha = 2.0$ and $\beta = 1.0$]{\includegraphics[scale=0.5]{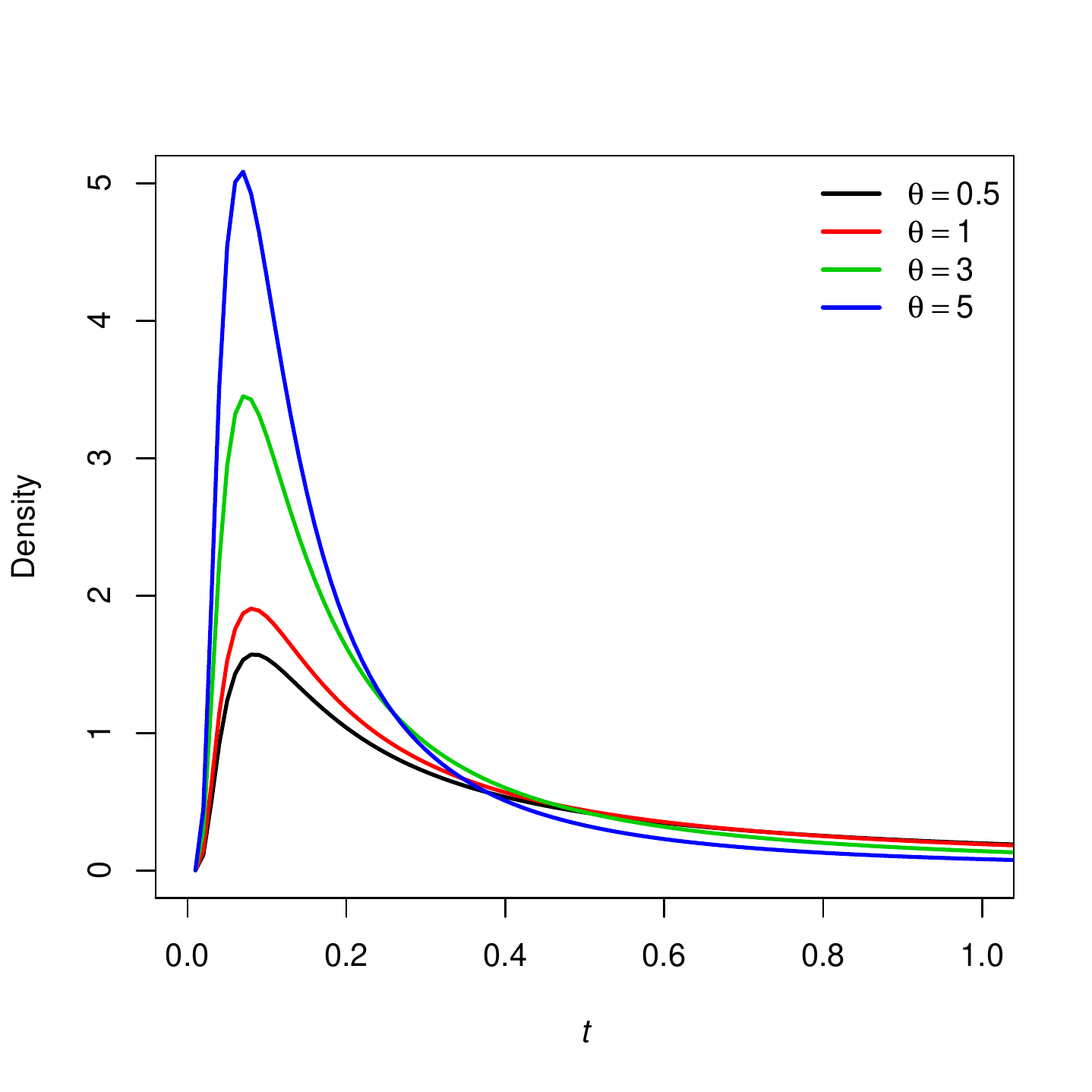}}
\subfigure[$\alpha = 0.2$ and $\beta = 1.0$]{\includegraphics[scale=0.5]{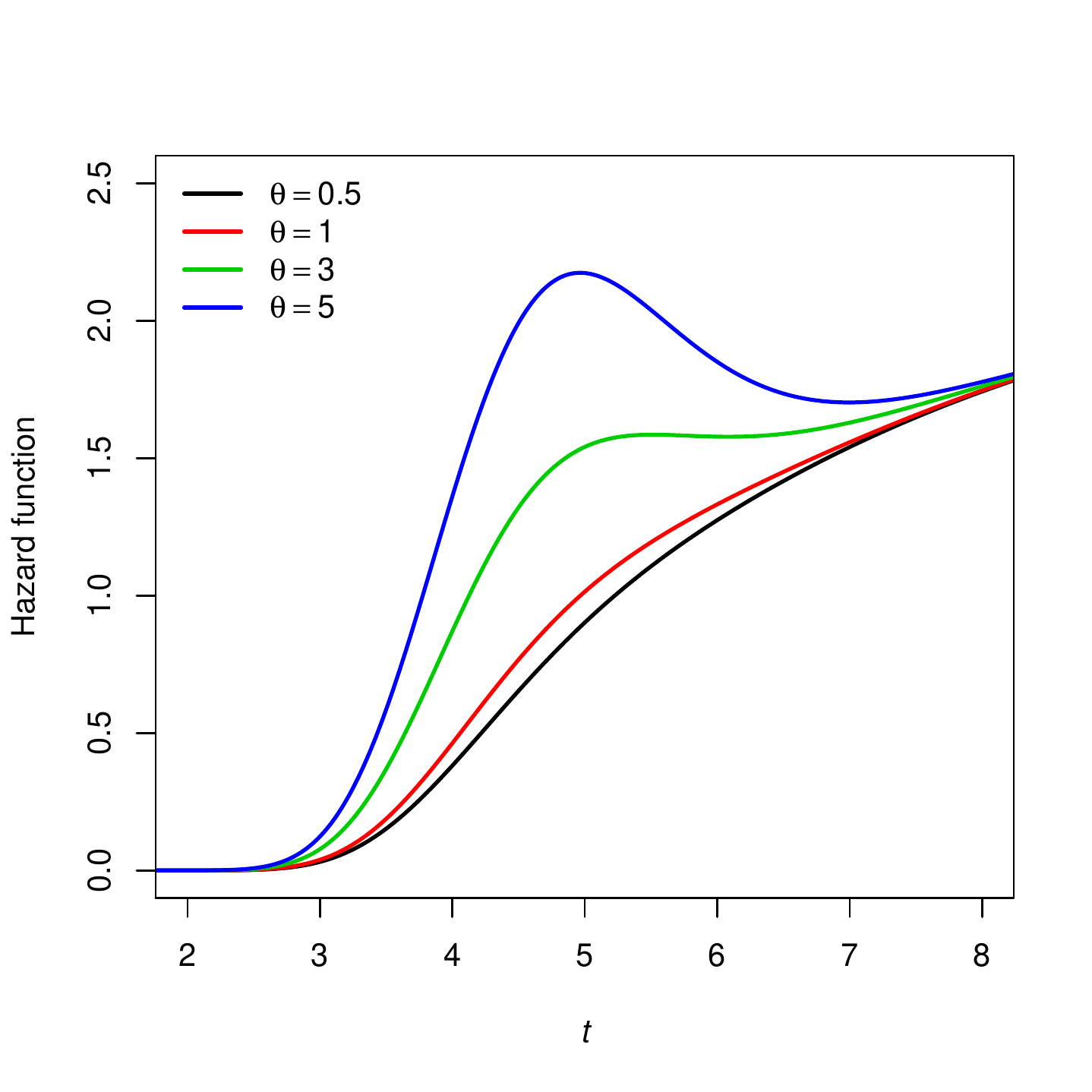}}
\subfigure[$\alpha = 2.0$ and $\beta = 1.0$]{\includegraphics[scale=0.5]{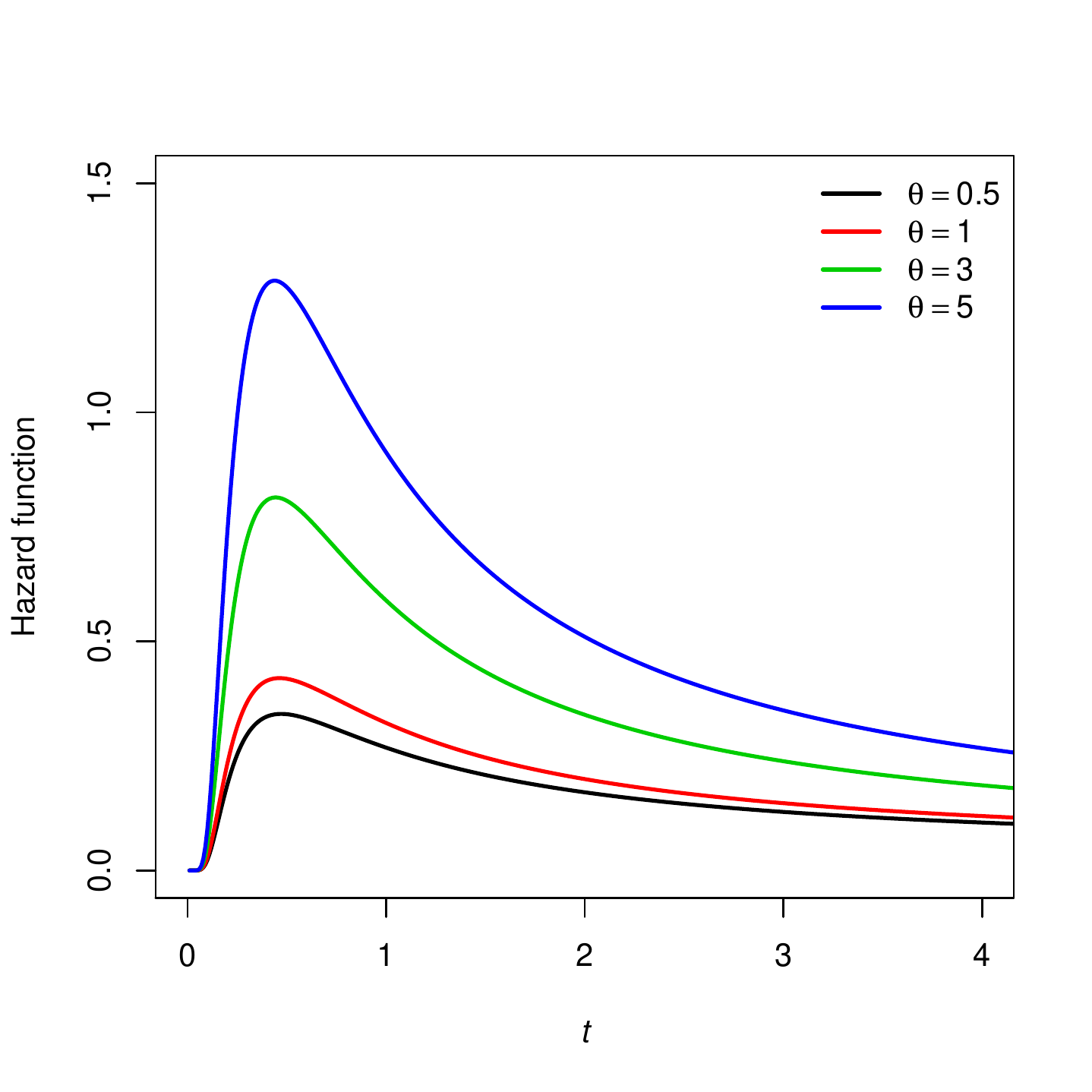}}
\caption{Plots of the $\mathcal{BSP}$ density and hazard rate functions for some parameter values.}
\label{figBSP}
\end{figure}

The expressions for the density function and the moments of the order statistics $X_{1:m}, \ldots, X_{m:m}$ from a random sample of the $\mathcal{BSP}$ distribution are
\begin{equation*}
f_{i:m}(x) = \frac{f_{\mathcal{BS}}(x; \alpha, \beta)}{B(i, m-i+1)} \sum_{k=0}^{i-1} (-1)^k \binom{i-1}{k} \left\{\frac{\exp\left[\theta\Phi(-\upsilon)\right]-1}{\mathrm{e}^\theta-1}\right\}^{m+k-i}
\end{equation*}
and
\begin{align*}
 \mathrm{E}(X^{s}_{i:m}) &= s \sum_{j=m-i+1}^m \frac{(-1)^{j-m+i-1}}{{(\exp(\theta)-1)^j}}\binom{j-1}{m-i}\binom{m}{j}\int_{0}^{\infty}x^{s-1} \left\{\exp\left[\theta\Phi(-\upsilon)\right]-1\right\}^j \mathrm{d}x.
\end{align*}

Let $x_1, \cdots, x_n$ be a random sample of size $n$ from $X\sim \mathcal{BSP}(\theta,\alpha,\beta)$. The log-likelihood function for the vector of parameters $\boldsymbol{\Theta} = (\theta,\alpha, \beta)^\top$ can be expressed as
\begin{eqnarray*}
\ell(\boldsymbol{\Theta}) &=& n\{\log(\theta) - \log[\exp(\theta) - 1] + \log[\kappa(\alpha, \beta)]\} - \frac{3}{2}\sum\limits_{i=1}^{n}\log(x_i) + \sum\limits_{i=1}^{n}\log(x_i + \beta)\\
&-& \frac{1}{2\alpha^2}\sum\limits_{i=1}^{n}\tau\left(\frac{x_i}{\beta}\right) + \theta\sum\limits_{i=1}^{n}[1-\Phi(v_i)].
\end{eqnarray*}

\subsection{Birnbaum-Saunders logarithmic distribution}

The cdf of the \textit{Birnbaum-Saunders logarithmic} ($\mathcal{BSL}$) distribution is defined by~\eqref{cdf} with $C(\theta) = -\log(1-\theta), \,\, 0 <\theta<1$, which corresponds to the logarithmic distribution. We obtain
\begin{equation*}
F_{\mathcal{BSL}}(x;\theta, \alpha, \beta) = 1 - \frac{\log\left[1-\theta(1-\Phi(\upsilon))\right]}{\log(1-\theta)}
\end{equation*}

The associated density and hazard rate functions are
\begin{equation*}
f_{\mathcal{BSL}}(x;\theta, \alpha, \beta) = -\frac{\theta f_{\mathcal{BS}(x)}}{\log(1-\theta)\left\{1 - \theta[1-\Phi(\upsilon)]\right\}}
\end{equation*}
and
\begin{equation*}
h_{\mathcal{BSL}}(x;\theta, \alpha, \beta) = -\frac{\theta f_{\mathcal{BS}}(x)}{\log\left\{1-\theta[1-\Phi(\upsilon)]\right\}\left\{1-\theta[1-\Phi(\upsilon)]\right\}}
\end{equation*}

\begin{Remark} The limit of the $\mathcal{BSL}$ hazard rate function as $x \rightarrow 0$ is $0$.
\end{Remark}

In Figure \ref{figBSL},  we plot the density and hazard rate functions of the $\mathcal{BSL}$ distribution for selected parameter values. We can verify that this distribution can have an upside-down bathtub or an increasing failure rate function depending on the values of its parameters.

\begin{figure}[!htbp]
\centering
\subfigure[$\alpha = 0.2$ and $\beta = 1.0$]{\includegraphics[scale=0.5]{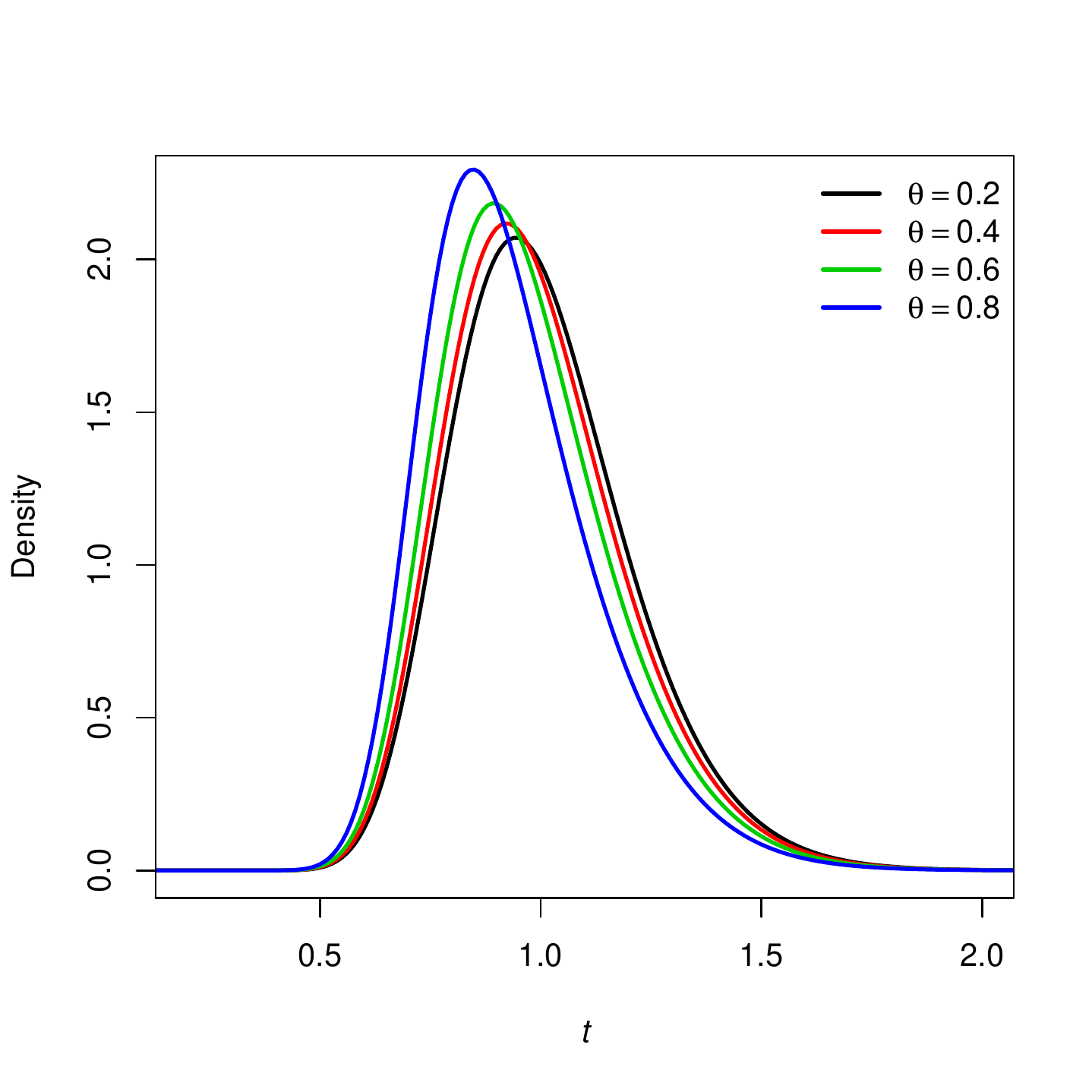}}
\subfigure[$\alpha = 2.0$ and $\beta = 1.0$]{\includegraphics[scale=0.5]{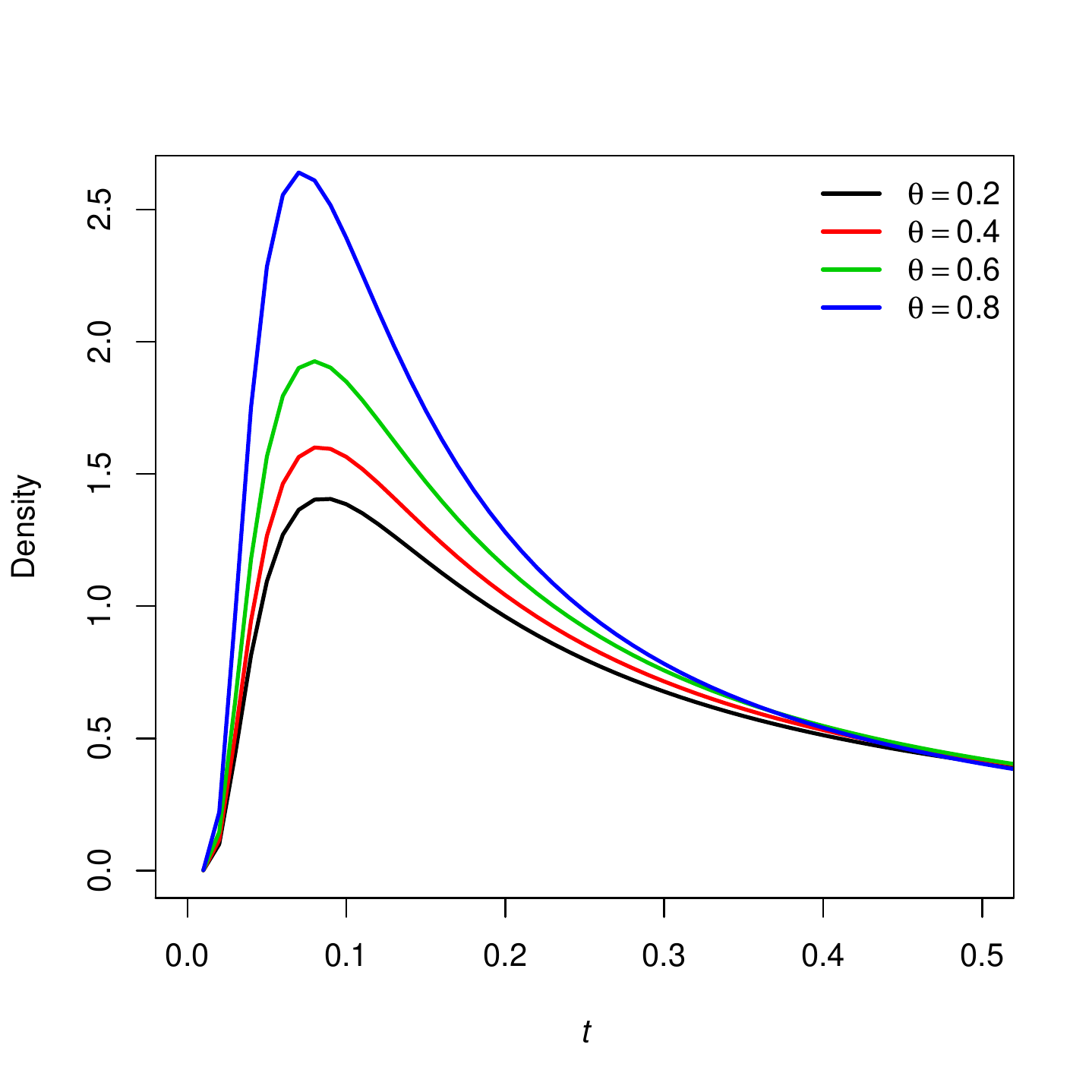}}
\subfigure[$\alpha = 0.2$ and $\beta = 1.0$]{\includegraphics[scale=0.5]{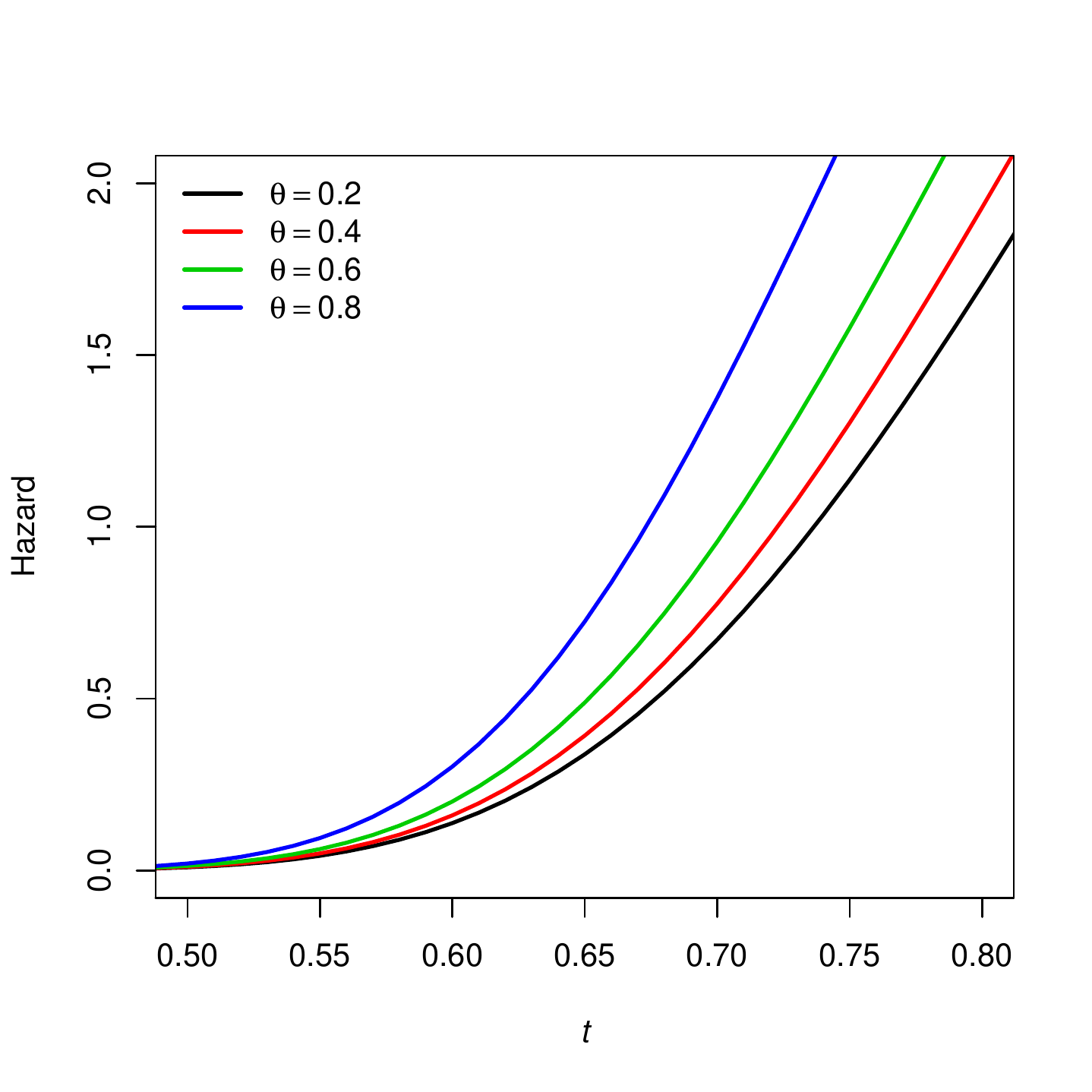}}
\subfigure[$\alpha = 2.0$ and $\beta = 1.0$]{\includegraphics[scale=0.5]{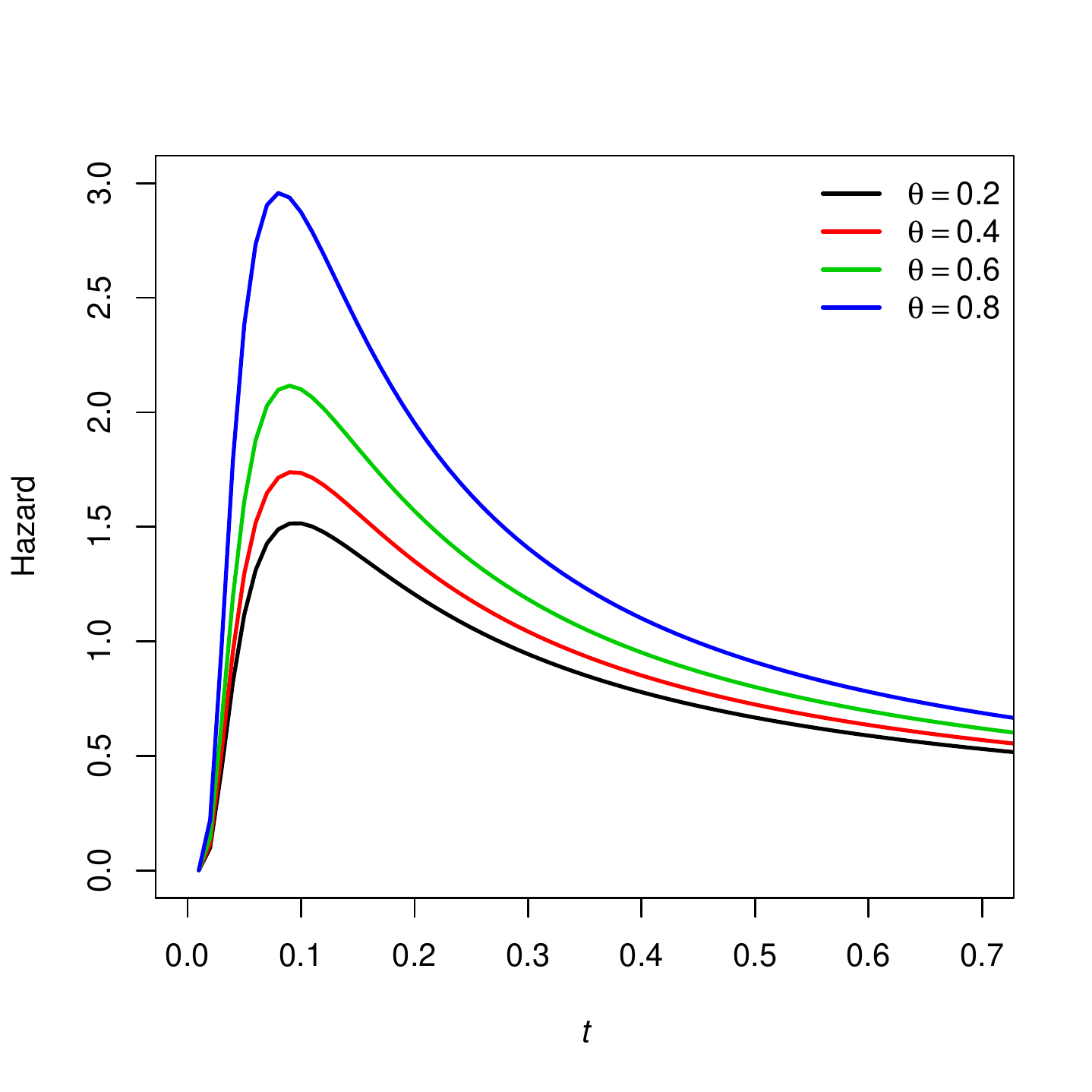}}
\caption{Plots of the $\mathcal{BSL}$ density and hazard rate functions for some parameter values.}
\label{figBSL}
\end{figure}

The expressions for the density function and the moments of the order statistics $X_{1:m}, \ldots, X_{m:m}$ from a random sample of the $\mathcal{BSL}$ distribution are
\begin{equation*}
f_{i:m}(x) = \frac{f_{\mathcal{BS}}(x; \alpha, \beta)}{B(i, m-i+1)} \sum_{k=0}^{i-1} (-1)^k \binom{i-1}{k} \left\{\frac{\log\left[1-\theta\Phi(-\upsilon)\right]}{\log(1-\theta)}\right\}^{m+k-i}
\end{equation*}
and
\begin{align*}
 \mathrm{E}(X^{s}_{i:m}) &= s \sum_{j=m-i+1}^m \frac{(-1)^{i-m-1}}{[\log(1-\theta)]^j}\binom{j-1}{m-i}\binom{m}{j}\int_{0}^{\infty}x^{s-1} \left\{\log\left[\frac{1}{1-\theta\Phi(-\upsilon)}\right]\right\}^{j} \mathrm{d}x.
\end{align*}

Let $x_1, \cdots, x_n$ be a random sample of size $n$ from $X\sim \mathcal{BSL}(\theta,\alpha,\beta)$. The log-likelihood function for the vector of parameters $\boldsymbol{\Theta} = (\theta,\alpha, \beta)^\top$ can be expressed as
\begin{eqnarray*}
\ell(\boldsymbol{\Theta}) &=& n\left\{\log(\theta) - \log\left[-\log(1-\theta)\right] + \log[\kappa(\alpha, \beta)]\right\} - \frac{3}{2}\sum\limits_{i=1}^{n}\log(x_i) + \sum\limits_{i=1}^{n}\log(x_i + \beta)\\
&-& \frac{1}{2\alpha^2}\sum\limits_{i=1}^{n}\tau\left(\frac{x_i}{\beta}\right) - \sum\limits_{i=1}^{n} \log\left\{1 - \theta[1-\Phi(\upsilon_i)]\right\}.
\end{eqnarray*}

\section{Applications}

In this section, we compare the fits of some special models of the $\mathcal{BSPS}$ class by means of two real data sets to show the potentiality of the new class. In order to estimate the parameters of these special models, we adopt the maximum likelihood method (as discussed in Section 4) and all the computations were done using the subroutine NLMixed of the SAS software. Obviously, due to the genesis of the $\mathcal{BS}$ distribution, the fatigue processes are by excellence ideally modeled by this distribution. Thus, the use of the $\mathcal{BS}$ and $\mathcal{BSPS}$ distributions for fitting these two data sets is well justified.

First, we consider a data set from Murthy \emph{et al}. (2004) consisting of the failure times of 20 mechanical components. The data are: 0.067, 0.068, 0.076, 0.081, 0.084, 0.085, 0.085, 0.086, 0.089, 0.098, 0.098, 0.114, 0.114, 0.115, 0.121, 0.125, 0.131, 0.149, 0.160, 0.485. The MLEs of the parameters (with corresponding standard errors in parentheses), the value of $-2\ell(\widehat{\boldsymbol{\Theta}})$, the Kolmogorov-Smirnov statistic, Akaike information criterion (AIC) and Bayesian information criterion (BIC) for the $\mathcal{BSG}$, $\mathcal{BSP}$, $\mathcal{BSL}$ and $\mathcal{BS}$ models are listed in Table \ref{tab1aplic}. Since the values of the AIC and BIC are smaller for the $\mathcal{BSG}$ distribution compared with those values of the other models, the new distribution seems to be a very  competitive model for these data.

\begin{table}[!htbp]
\centering
\caption{Parameter estimates, K-S statistics, AIC and BIC for failure times data}\label{tab1aplic}
\scalebox{0.87}[0.87]{
\begin{threeparttable}
\renewcommand{\arraystretch}{1.3}
\begin{tabular}{lccccccccc}
\toprule
Distribution    &$\widehat{\alpha}$    &$\widehat{\beta}$ &$\widehat{\theta}$  &&K--S&&$-2\ell(\widehat{\boldsymbol{\Theta}})$  &AIC    &BIC \\
\hline
$\mathcal{BSG}$ &0.6461        &0.4521    &0.9950     & &0.1314  &        &$-$77.6 &$-$71.6 &$-$68.6\\
                &(0.6194)\tnote{a}      &(0.8247)  &(0.0184)     & &  &        &        &        &\\
$\mathcal{BSP}$ &0.4774        &0.1735    &5.1057     & &0.1224  &        &$-$73.2 &$-$67.2 &$-$64.3\\
                &(0.0877)\tnote{a}      &(0.0304)  &(2.0932)     & &  &        &       &      &\\
$\mathcal{BSL}$ &0.3549        &0.2437    &0.9999     & &0.2055  &        &$-$74.0      &$-$68.0      &$-$65.0 \\
                &(0.0570)\tnote{a}            &(0.0419)  &(0.0001)   & &  &        &       &      &\\
$\mathcal{BS}$  &0.4466        &0.1107    &           & &0.2029  &        &$-$65.5       &$-$61.5  &$-$59.5
 \\
                &(0.0706)\tnote{a}  &(0.0108)  &     & &  &        &       &      &\\
 \hline
\end{tabular}
\begin{tablenotes}
       \item[a] Denotes the standard deviation of the MLE's of $\alpha, \beta$ and $\theta$.
\end{tablenotes}
\end{threeparttable}}
\end{table}

Plots of the pdf and cdf of the $\mathcal{BSG}$, $\mathcal{BSP}$, $\mathcal{BSL}$ and $\mathcal{BS}$ fitted models to these data are displayed in Figure \ref{fig:cdfplot}. They indicate that the $\mathcal{BSG}$ distribution is superior to the other distributions in terms of model fitting.
Additionally, it is evident that the $\mathcal{BS}$ distribution presents the worst fit to the current data and then the proposed models
outperform this distribution.

\begin{figure}[!htbp]
\centering
\subfigure[]{\includegraphics[scale=0.53]{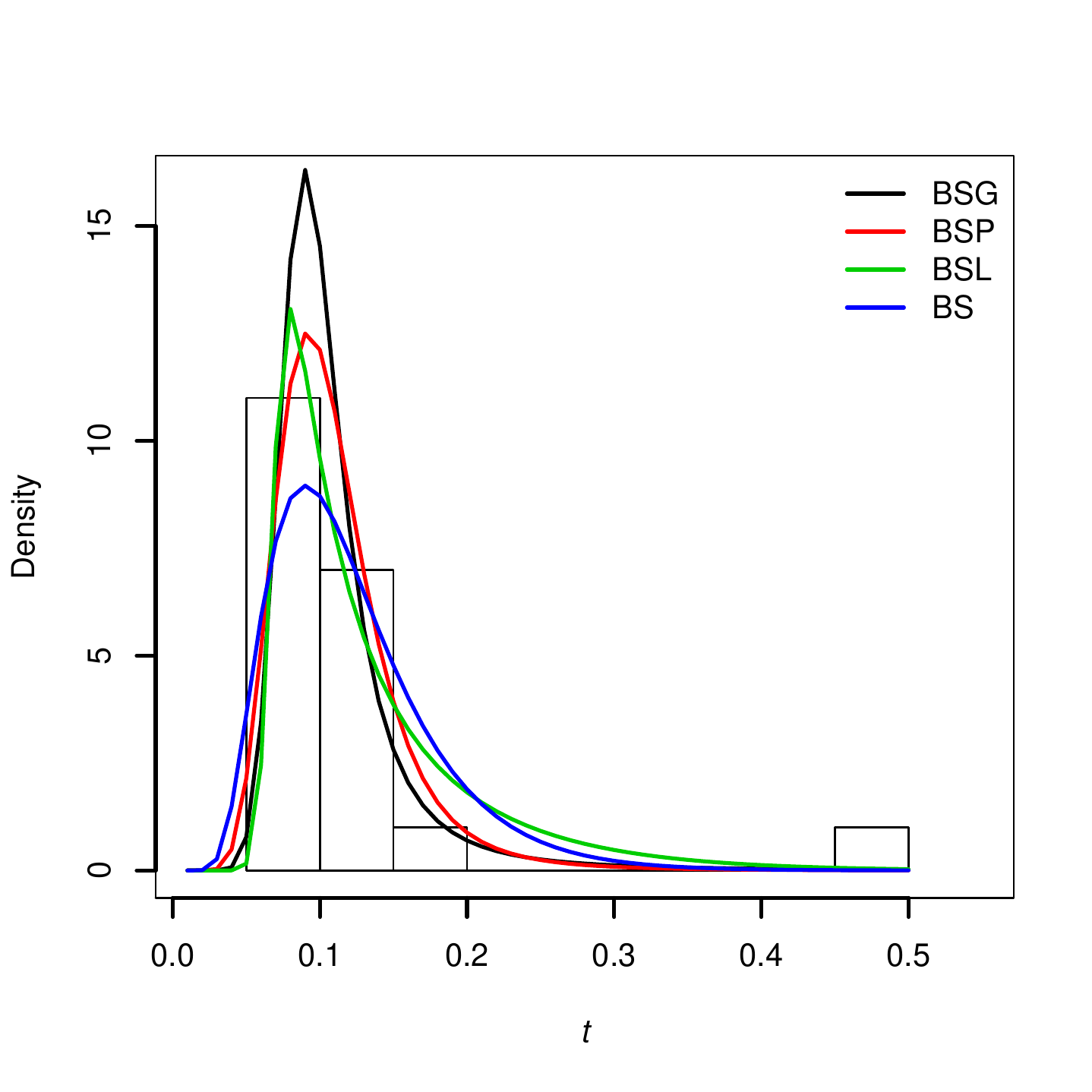}}
\subfigure[]{\includegraphics[scale=0.53]{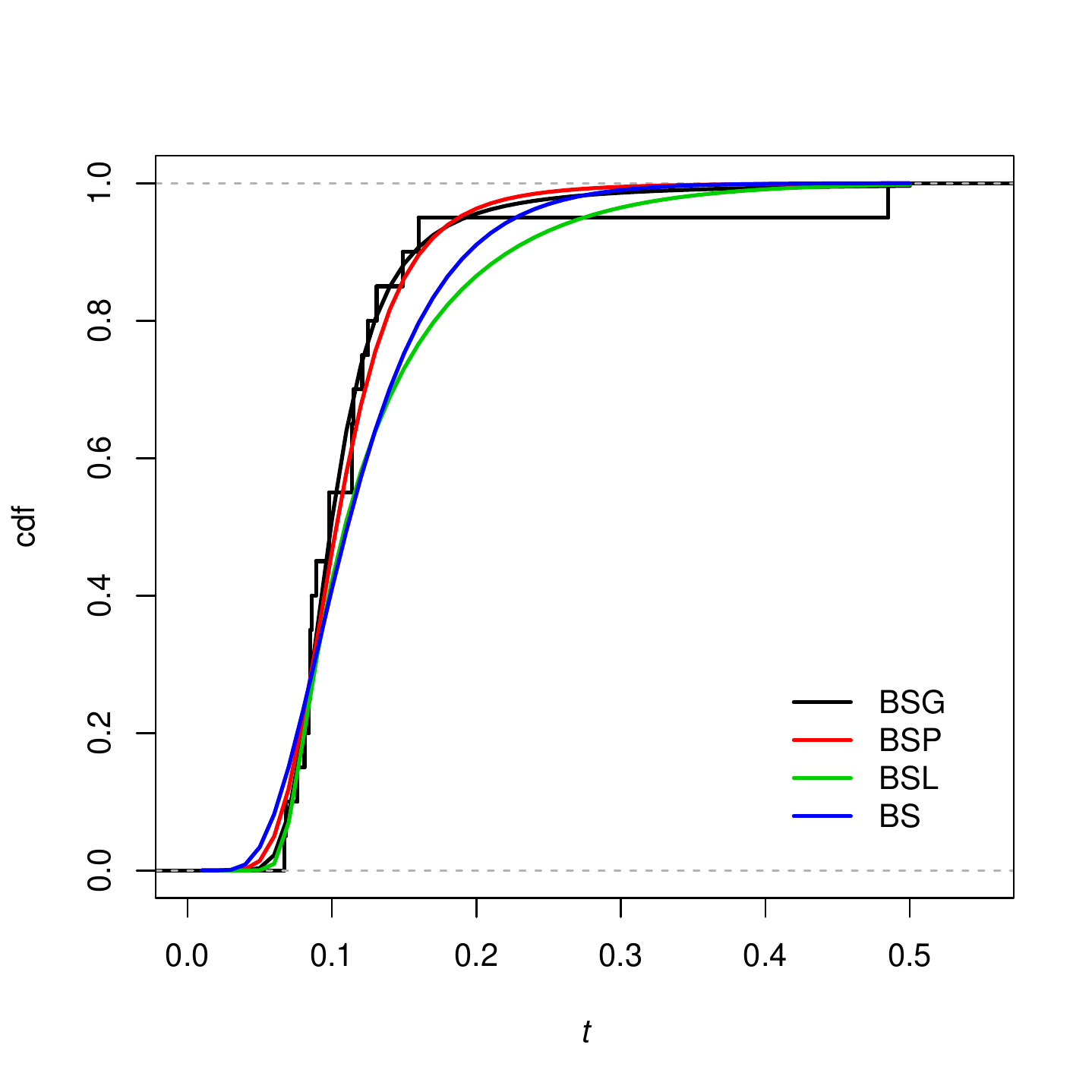}}
\caption{Estimated (a)  pdf and (b) cdf for the $\mathcal{BSG}$, $\mathcal{BSP}$, $\mathcal{BSL}$ and $\mathcal{BS}$ models for failure times data.}
\label{fig:cdfplot}
\end{figure}

We can also perform formal goodness-of-fit tests in order to verify which distribution fits better to these data. Chen and Balakrishnan (1995) proposed a general approximate goodness-of-fit test for the hypothesis $\mathcal{H}_0: X_1, \ldots, X_n$ with $X_i$ following $F(x; \theta)$, i.e., under $\mathcal{H}_0, X_1, \ldots, X_n$ is a random sample from a continuous distribution with cumulative distribution $F(x; \theta)$, where the form of $F$ is known but the $p$-vector $\boldsymbol{\theta}$ is unknown. The method is based on the Cramér-von Mises (C--M) and Anderson-Darling (A--D) statistics and, in general, the smaller the values of these statistics, the better the fit.

Table \ref{tab2aplic} gives the values of the C--M and A--D statistics ($p$-values between parentheses) for the failure times data. Thus, according to these formal tests, the $\mathcal{BSG}$ model fits the current data better than the other models. These results illustrate the potentiality of the $\mathcal{BSG}$ distribution and the importance of the additional parameter.

\begin{table}[!htbp]
\centering
\small
\caption{Goodness-of-fit tests}\label{tab2aplic}
\begin{threeparttable}
\renewcommand{\arraystretch}{1.3}
\begin{tabular}{lccccccc}
\hline
\multirow{2}{*}{Model}&\multicolumn{6}{c}{Statistics}\\
\cline{2-7}
                && C--M && &A--D \\
\hline
$\mathcal{BSG}$ && 0.0469 (0.5563)\tnote{a}&& &0.3116 (0.5517)\tnote{a}\\
$\mathcal{BSP}$ && 0.0877 (0.1650)&& &0.6629 (0.0835)\\
$\mathcal{BSL}$ && 0.0784 (0.2173)&& &0.6064 (0.1151) &\\
$\mathcal{BS}$  && 0.1967 (0.0059)&& & 1.3748 (0.0015)\\
\hline
	\end{tabular}
		\begin{tablenotes}
       \item[a] Denotes the $p$-value of the test.
		\end{tablenotes}
	\end{threeparttable}
\end{table}

As a second application, we shall analyze a data set from McCool (1974) on the fatigue lives (in hours) of 10 bearings of a certain kind. The data are: 152.7, 172.0, 172.5, 173.3, 193.0, 204.7, 216.5, 234.9, 262.6, 422.6. Table \ref{tabappli3} gives the MLEs (standard errors between parentheses) of the parameters, the values of $-2\ell(\widehat{\boldsymbol{\Theta}})$ and of the Kolmogorov-Smirnov, AIC and BIC statistics for the fitted $\mathcal{BSG}$, $\mathcal{BSP}$ and $\mathcal{BS}$ models to the second data set. Plots of the fitted $\mathcal{BSG}$, $\mathcal{BSP}$ and $\mathcal{BS}$ densities are given in Figure \ref{fig:cdfplot2}.

\begin{table}[!htbp]
\centering
\caption{Parameter estimates, K-S statistics, AIC and BIC for fatigue life data}\label{tabappli3}
\scalebox{0.87}[0.87]{
\begin{threeparttable}
\renewcommand{\arraystretch}{1.3}
\begin{tabular}{lccccccccc}
\toprule
Distribution    &$\widehat{\alpha}$    &$\widehat{\beta}$ &$\widehat{\theta}$  &&K--S&&$-2\ell(\widehat{\boldsymbol{\Theta}})$  &AIC    &BIC \\
\hline
$\mathcal{BSG}$ &0.3087        &350.98    &0.9672     & &0.1681    & &106.9 &112.9 &113.8\\
                &(0.1285)\tnote{a}      &(182.50)  &(0.0861)   & &          & &      &      &\\
$\mathcal{BSP}$ & 0.2917       &259.20    &3.1140     & &0.1633    & &108.3 &114.3 &115.2\\
                &(0.0772)\tnote{a}       &(44.4148) &(2.4589)   & &          & &      &      &\\
$\mathcal{BS}$  &0.2825        &212.05    &           & &0.1707    & &109.9 &113.9 &114.5 \\
                &(0.0632)\tnote{a}       &(18.7530) &           & &          &        &      &\\
 \hline
\end{tabular}
\begin{tablenotes}
       \item[a] Denotes the standard deviation of the MLE's of $\alpha, \beta$ and $\theta$.
\end{tablenotes}
\end{threeparttable}}
\end{table}

\begin{figure}[!htbp]
\centering
\subfigure[]{\includegraphics[scale=0.5]{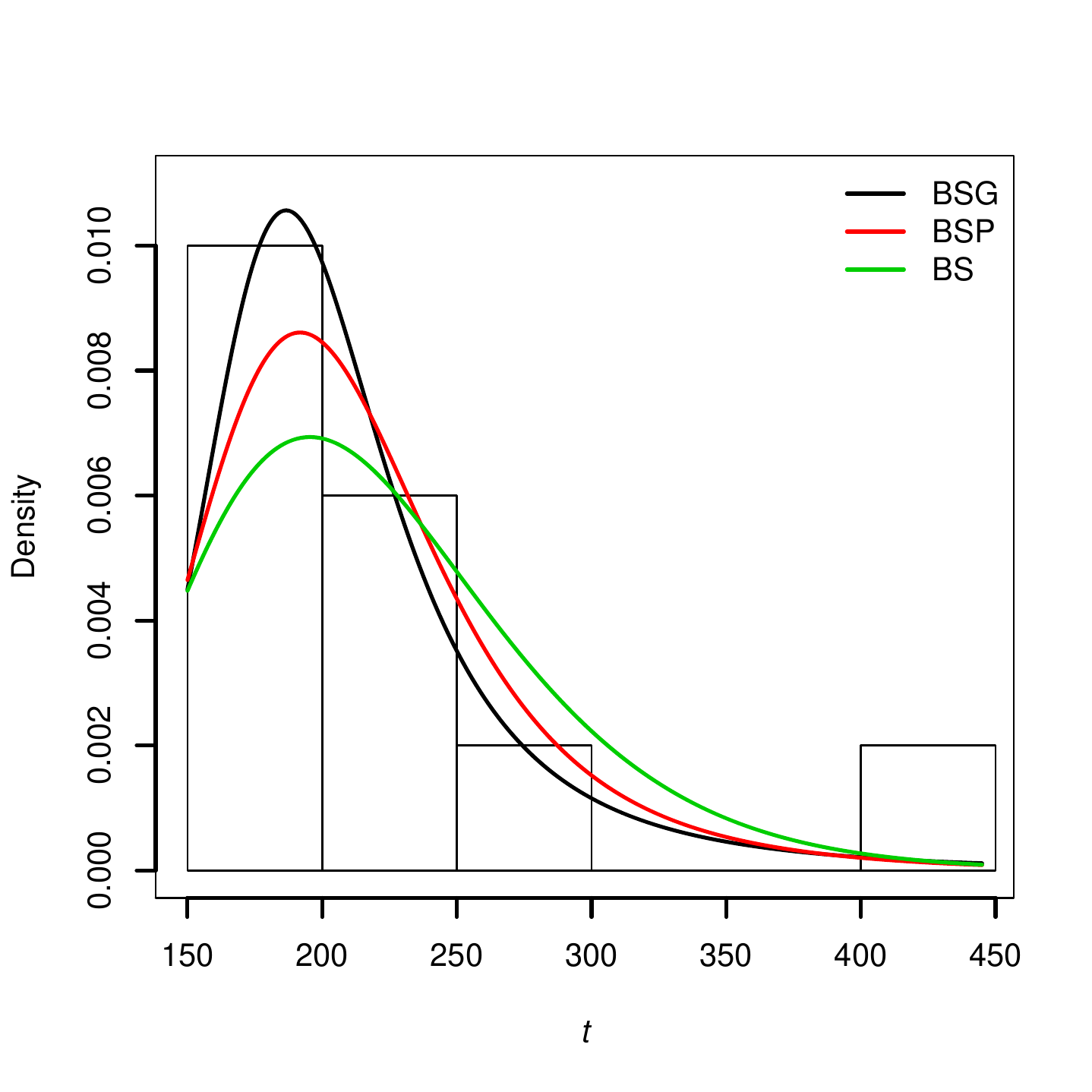}}
\subfigure[]{\includegraphics[scale=0.5]{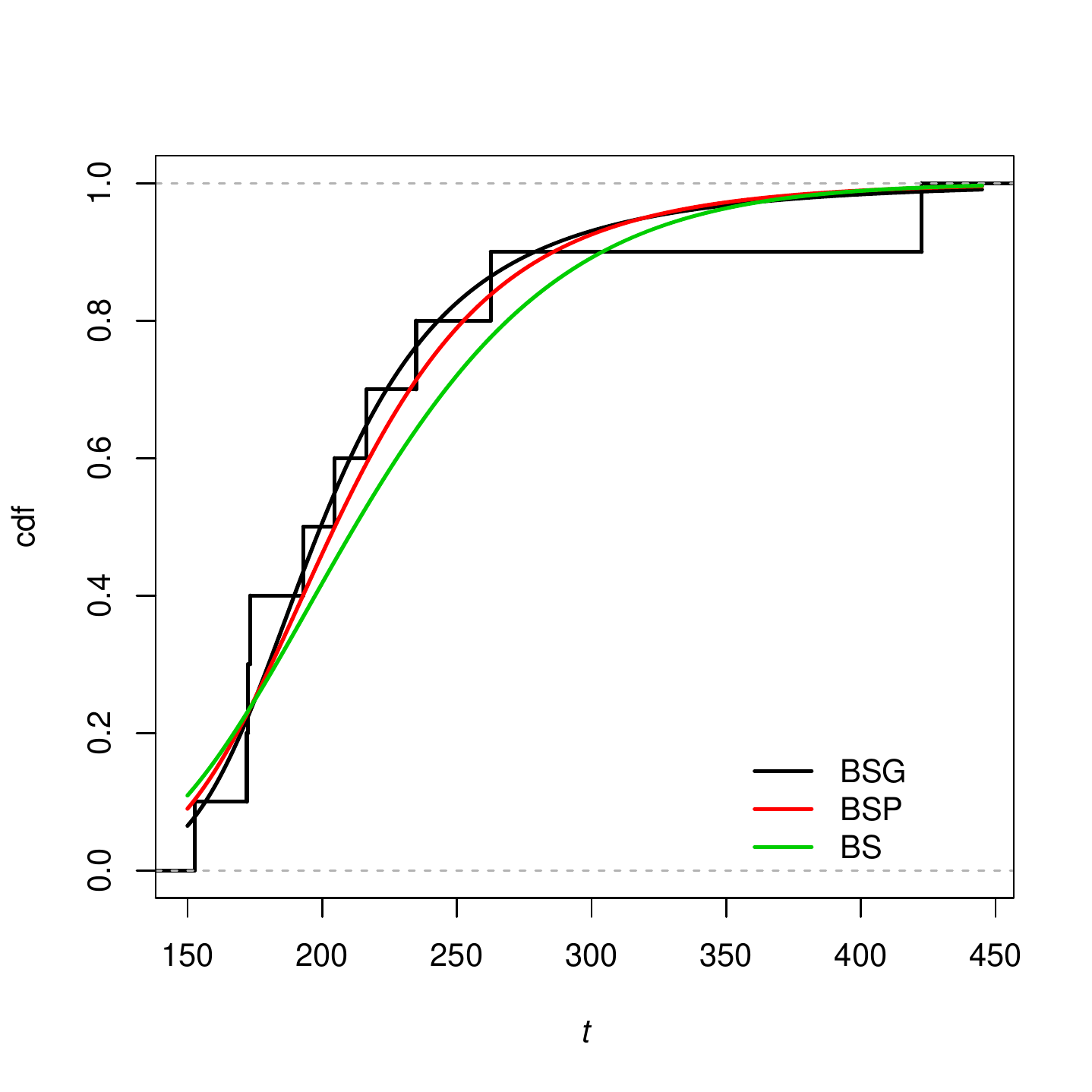}}
\caption{Estimated (a)  pdf and (b) cdf for the $\mathcal{BSG}$, $\mathcal{BSP}$ and $\mathcal{BS}$ models for fatigue lives data.}
\label{fig:cdfplot2}
\end{figure}

The figures in Table \ref{tabappli3} indicate that the $\mathcal{BSG}$ model yields a better fit to these data than the other models.
Table \ref{tab4aplic} gives the values of the C--M and A--D statistics ($p$-values between parentheses) for the fatigue life data. So, these values indicate that the null hypothesis is strongly not rejected for the $\mathcal{BSG}$ distribution, whereas the null hypothesis is weakly not rejected for the other models.
Thus, according to these goodness-of-fit tests, the $\mathcal{BSG}$ model fits the current data better than the other models.

\begin{table}[!htbp]
\centering
\small
\caption{Goodness-of-fit tests}\label{tab4aplic}
\begin{threeparttable}
\renewcommand{\arraystretch}{1.3}
\begin{tabular}{lccccccc}
\hline
\multirow{2}{*}{Model}&\multicolumn{6}{c}{Statistics}\\
\cline{2-7}
                && C--M && &A--D \\
\hline
$\mathcal{BSG}$ && 0.0370 (0.7373)\tnote{a}&&& 0.2761 (0.6575)\tnote{a}\\
$\mathcal{BSP}$ && 0.0573 (0.4108)& &&0.4243 (0.3178) \\
$\mathcal{BS}$  && 0.0862 (0.1725)&& &0.6148 (0.1098)& \\
\hline
	\end{tabular}
		\begin{tablenotes}
       \item[a] Denotes the $p$-value of the test.
		\end{tablenotes}
	\end{threeparttable}
\end{table}
\section{Concluding remarks}
The two-parameter Birnbaum-Saunders distribution is extended to a new class of fatigue life distributions by compounding the Birnbaum-Saunders and power series distributions, thus creating the Birnbaum-Saunders power series distribution. As expected, the new class includes as special model the $\mathcal{BS}$ distribution. The $\mathcal{BSPS}$ failure rate function can be increasing and upside-down bathtub shaped. Some mathematical properties of the new class are provided such as the ordinary moments, quantile function, order statistics and their moments. The estimation of the model parameters is approached by the method of maximum likelihood and the observed information matrix is derived. We fit some $\mathcal{BSPS}$ distributions to two real data sets to show the potentiality of the proposed class.

Future research could be adressed to study the complementary Birnbaum-Saunders power series distribution. This class of distributions is a suitable model in a complementary risk problem based in the presence of latent risks which arise in several areas such as public health, actuarial science, biomedical studies, demography and industrial reliability.

\section*{Acknowledgement}

The financial support from the Brazilian governmental institutions CNPq and CAPES is gratefully acknowledged.

\appendix
\section{Elements of the observed information matrix}
The observed information matrix $J(\boldsymbol{\Theta})$ for the parameters $\theta, \alpha$ and $\beta$ is given by
$$J(\boldsymbol{\Theta})= -\frac{\partial^2 \ell(\boldsymbol{\Theta})}{\partial \boldsymbol{\Theta} \partial \boldsymbol{\Theta}^\top} = \left(
\begin{array}{cccc}
U_{\theta\theta}&U_{\theta \alpha}&U_{\theta \beta}\\
\cdot&U_{\alpha \alpha}&U_{\alpha \beta}\\
\cdot&\cdot&U_{\beta \beta}\\
\end{array}\right),
$$
whose elements are
\begin{align*}
U_{\theta \theta} &= -\frac{n}{\theta^2} - n \left[\frac{C''(\theta)}{C(\theta)} - \left(\frac{C'(\theta)}{C(\theta)}\right)^2\right] -
	\sum_{i=1}^n \Phi(-\upsilon_i)^2 \left[\frac{C''\left(-\theta\Phi(-\upsilon_i)\right)}{C'\left(-\theta\Phi(-\upsilon_i)\right)}\right]^2,\\
U_{\theta \alpha} & = \frac{1}{\alpha^2} \sum_{i=1}^n \phi(\upsilon_i) \rho(t_i/\beta) \frac{C'_\alpha\left(-\theta \Phi(-\upsilon)\right)}{C'\left(-\theta \Phi(-\upsilon)\right)} + \frac{\theta}{\alpha^2} \sum_{i=1}^n \phi(\upsilon_i) \tau(\sqrt{t_i/\beta}) \frac{C'_{\alpha \beta}\left(-\theta \Phi(-\upsilon)\right)}{C'\left(-\theta \Phi(-\upsilon)\right)}\\ &- \frac{\theta}{\alpha^2}\sum_{i=1}^n\phi(\upsilon_i)\Phi(-\upsilon_i)\tau(\sqrt{t_i/\beta})\frac{C'_{\alpha}\left(-\theta \Phi(-\upsilon)\right)}{C'\left(-\theta \Phi(-\upsilon)\right)},\\
U_{\theta \beta} & = -\frac{1}{2\alpha \beta}\sum_{i=1}^n \frac{C'_{\beta}\left(-\theta \Phi(-\upsilon)\right)}{C'\left(-\theta \Phi(-\upsilon)\right)} -
\frac{\theta}{2\alpha \beta} \sum_{i=1}^n \frac{C'_{\theta \beta}\left(-\theta \Phi(-\upsilon)\right)}{C'\left(-\theta \Phi(-\upsilon)\right)} +
\frac{\theta}{2 \alpha \beta} \sum_{i=1}^n \Phi(-\upsilon_i)C'_{\alpha}\left(-\theta \Phi(-\upsilon)\right)\frac{C''\left(-\theta \Phi(-\upsilon)\right)}{C'\left(-\theta \Phi(-\upsilon)\right)},\\
U_{\alpha \alpha} & = -\frac{n}{\alpha^2}\left(1+\frac{6}{\alpha^2}\right) - \frac{3}{\alpha^4}\sum_{i=1}^n \tau\left(\frac{t_i}{\beta}\right)
 - \frac{2\theta}{\alpha^3} \sum_{i=1}^n \frac{C'_{\alpha}\left(-\theta \Phi(-\upsilon)\right)}{C'\left(-\theta \Phi(-\upsilon)\right)} +
\frac{\theta}{\alpha^2} \sum_{i=1}^n \frac{C'_{\alpha \alpha}\left(-\theta \Phi(-\upsilon)\right)}{C'\left(-\theta \Phi(-\upsilon)\right)}\\
&- \left(\frac{\theta}{\alpha^2}\right)^2 \sum_{i=1}^n \phi(\upsilon_i)\rho\left(\frac{t_i}{\beta}\right)C'_{\alpha}\left(-\theta \Phi(-\upsilon)\right)\frac{C''\left(-\theta \Phi(-\upsilon)\right)}{C'\left(-\theta \Phi(-\upsilon)\right)^2},\\
U_{\alpha \beta} & = -\frac{1}{\alpha^3 \beta} \sum_{i=1}^n \left(\frac{t_i}{\beta} - \frac{\beta}{t_i}\right) - \frac{\theta}{2 \alpha^3 \beta} \sum_{i=1}^n \upsilon_i \phi(\upsilon_i) \left(\sqrt{\frac{t_i}{\beta}} + \sqrt{\frac{\beta}{t_i}}\right) C'_{\alpha}\left(-\theta \Phi(-\upsilon)\right) \frac{C''\left(-\theta \Phi(-\upsilon)\right)}{C'\left(-\theta \Phi(-\upsilon)\right)} \\
&+\frac{\theta}{\alpha^2} \sum_{i=1}^n \frac{C'_{\alpha \beta}\left(-\theta \Phi(-\upsilon)\right)}{C'\left(-\theta \Phi(-\upsilon)\right)}, \\
\end{align*}

\begin{align*}
U_{\beta \beta} & = \frac{n}{2 \beta^2} - \frac{1}{2 \alpha^2 \beta^2}\sum_{i=1}^n \left(\frac{t_i}{\beta} + \frac{\beta}{t_i}\right) - \frac{1}{2 \alpha^2 \beta^2} \sum_{i=1}^n \tau\left(\frac{t_i}{\beta}\right) - \frac{\theta}{2 \alpha \beta^2} \sum_{i=1}^n \frac{C'_{\beta}\left(-\theta \Phi(-\upsilon)\right)}{C'\left(-\theta \Phi(-\upsilon)\right)}\\
&- \frac{\theta^2}{2 \alpha^2 \beta} \sum_{i=1}^n \phi(\upsilon_i)^2 \tau\left(\sqrt{\frac{t_i}{\beta}}\right)\left(\frac{t_i}{2 \beta^2}\sqrt{\frac{\beta}{t_i}} - \frac{1}{2 t_i}\sqrt{\frac{t_i}{\beta}}\right)\frac{C'_{\beta}\left(-\theta \Phi(-\upsilon)\right)}{C'\left(-\theta \Phi(-\upsilon)\right)^2}\\
&+ \frac{\theta}{(2\alpha\beta)^2} \sum_{i=1}^n \upsilon_i \phi(\upsilon_i) \tau\left(\sqrt{\frac{t_i}{\beta}}\right) \left(\sqrt{\frac{t_i}{\beta}} + \sqrt{\frac{\beta}{t_i}}\right) \frac{C'_{\beta}\left(-\theta \Phi(-\upsilon)\right)}{C'\left(-\theta \Phi(-\upsilon)\right)} + \frac{\theta}{4 \alpha \beta^2} \sum_{i=1}^n \phi(\upsilon_i) \rho\left(\sqrt{\frac{t_i}{\beta}}\right)\frac{C'_{\beta}\left(-\theta \Phi(-\upsilon)\right)}{C'\left(-\theta \Phi(-\upsilon)\right)}\\
&+ \frac{\theta}{2 \alpha \beta}\sum_{i=1}^n \phi(\upsilon_i) \tau\left(\sqrt{\frac{t_i}{\beta}}\right) \frac{C'_{\beta \beta}\left(-\theta \Phi(-\upsilon)\right)}{C'\left(-\theta \Phi(-\upsilon)\right)},
\end{align*}
where $C'_{k}(\cdot)$ is the first derivative of $C'(\cdot)$ with respect to $k$ and $C'_{jk}(\cdot)$ corresponds to second derivative with respect to $j$ and $k$.

\end{document}